%% file: AA_main.tex
\documentclass[11pt,letterpaper]{article}

\usepackage{hyperref}
\hypersetup{
    colorlinks=true,
    linkcolor=teal,
    filecolor=magenta,
    urlcolor=cyan,
    citecolor=teal
}

\usepackage[margin=1in]{geometry}
\usepackage[ruled,vlined,linesnumbered]{algorithm2e}
\usepackage{amssymb}
\usepackage{xcolor, csquotes, amsmath, amsthm, hyperref, enumitem, amsfonts, multirow, mathtools, footnote, etoolbox, csquotes, mathrsfs, threeparttable, nicefrac, makecell, multirow, cite}
\usepackage[capitalize]{cleveref}
\usepackage{tocloft}
\usepackage{cleveref}

\newcommand{\appendixbreak}{\addtocontents{toc}{\vspace{1ex}}}

\definecolor{teal}{RGB}{60, 100, 200}
\usepackage[all]{hypcap}
\hypersetup{
    colorlinks=true,
    linkcolor=teal,
    filecolor=magenta,
    urlcolor=cyan,
    citecolor=teal,
    linktoc=all,
    pdfstartview=FitH,
    pdfpagelayout=SinglePage,
    pdfnewwindow=true,
    pdfdisplaydoctitle=true
}

\let\oldFootnote\footnote
\newcommand\nextToken\relax
\renewcommand\footnote[1]{\oldFootnote{#1}\futurelet\nextToken\isFootnote}
\newcommand\isFootnote{\ifx\footnote\nextToken\textsuperscript{,}\fi}

\newtheorem*{theorem*}{Theorem}
\newtheorem*{corollary*}{Corollary}
\newtheorem{observation}{Observation}

\newtheorem{corollary}{Corollary}
\newtheorem{lemma}{Lemma}
\newtheorem{theorem}{Theorem}
\newtheorem{definition}{Definition}

\newtheorem{proposition}{Proposition}
\newtheorem{question}{Question}
\newtheorem{fact}{Fact}
\theoremstyle{remark} 
\newtheorem{remark}{Remark}

\numberwithin{observation}{section}
\numberwithin{claim}{section}
\numberwithin{corollary}{section}
\numberwithin{lemma}{section}
\numberwithin{theorem}{section}
\numberwithin{definition}{section}
\numberwithin{notation}{section}
\numberwithin{note}{section}
\numberwithin{notice}{section}
\numberwithin{proposition}{section}
\numberwithin{fact}{section}

\definecolor{greenish}{RGB}{0, 140, 80}
\newcommand{\gcm}[1]{{\color{greenish} [#1]}}

\renewcommand{\ldots}{\ldotp\ldotp}
\newcommand{\al}{a}
\newcommand{\be}{b}
\newcommand{\aL}{a}
\newcommand{\bL}{b}
\newcommand{\pay}{\rho}

\title{Welfare and Beyond in Multi-Agent Contracts\thanks{First version: November 2024; this version: April 2025.}}

\author{Gil Aharoni\thanks{School of Computer Science, Tel Aviv University. Email: \texttt{gilaroni35@gmail.com}.} \and
Martin Hoefer\thanks{Dept.~of Computer Science, RWTH Aachen University. Email: \texttt{mhoefer@cs.rwth-aachen.de.}} \and
Inbal Talgam-Cohen\thanks{School of Computer Science, Tel Aviv University. Email: \texttt{inbaltalgam@gmail.com}.}}

\SetKwInput{KwInput}{Input}
\SetKwInput{KwOutput}{Output}

\newcommand{\bfIf}{\textbf{if}}
\newcommand{\bfThen}{\textbf{then}}
\newcommand{\bfElse}{\textbf{else}}

\date{} 

\begin{document}
\pagenumbering{gobble}  
\begin{titlepage}

\maketitle
\input{BA_abstract}
\setcounter{tocdepth}{2} 
{\hypersetup{linkcolor=black} \vfill \tableofcontents}

\end{titlepage}

\pagenumbering{arabic}  
\input{BB_start}
\input{CC_model}

\input{CD_technical-overview}
\input{FF_ratios}
\input{DA_alg3}

\input{GG_sublinear}

\input{AZ_value} 
\input{AZ_acknowledgement}

\bibliographystyle{abbrv}
\bibliography{ZZbib.bib}

\appendix
\appendixbreak

\input{KK_appendix}

\end{document}

%% file: BA_abstract.tex
\begin{abstract}

A principal delegates a project to a team $S$ from a pool of $n$ agents. The project's value if all agents in $S$ exert costly effort is $f(S)$. To incentivize the agents to participate, the principal assigns each agent $i\in S$ a share $\rho_i\in [0,1]$ of the project's final value (i.e., designs $n$ linear contracts). The shares must be feasible---their sum should not exceed $1$. It is well-understood how to design these contracts to maximize the principal's own expected utility, but what if the goal is to coordinate the agents toward maximizing social welfare? 

We initiate a systematic study of multi-agent contract design with objectives beyond principal's utility, including  welfare maximization, for various classes of value functions $f$. Our exploration reveals an arguably surprising fact: If $f$ is up to XOS in the complement-free hierarchy of functions, then the optimal principal's utility is a constant-fraction of the optimal welfare. This is in stark contrast to the much larger welfare-utility gaps in auction design, and no longer holds above XOS in the  hierarchy, where the gap can be unbounded. 

A constant bound on the welfare-utility gap immediately implies that existing algorithms for designing contracts with approximately-optimal principal's utility also guarantee approximately-optimal welfare. The downside of reducing welfare to utility is the loss of large constants. To obtain better guarantees, we develop polynomial-time algorithms directly for welfare, for different classes of value functions. These include a tight $2$-approximation to the optimal welfare for symmetric XOS functions. 

Finally, we extend our analysis beyond welfare to the project's value under general feasibility constraints. Our results immediately translate to budgeted welfare and utility. 
\end{abstract}

%% file: BB_start.tex
\section{Introduction}

In contract design, a principal delegates a project to one or more agents. The contract specifies the principal's payments to the agents according to the project's outcome.
These payments have two effects: (i)~They shape the agents' efforts, and thus, how much welfare is created from their work on the project; (ii)~They divide the welfare among the principal and agents. 

The rapidly-expanding body of research on \emph{combinatorial} contracts aims to find contracts with good performance in polynomial time, given settings with a combinatorial flavor (such as many possible agent combinations, or multi-dimensional project outcomes or agent actions). So far, the focus has been on the objective of maximizing the \emph{principal's utility} (also referred to simply as \emph{utility}). In combinatorial contract settings with a \emph{single} principal-agent pair~\cite{DuttingRT20,DuettingEFK21}, maximizing the principal's utility is indeed the main objective of interest. This is because the other natural objective of maximizing \emph{welfare} is straightforward to achieve: the contract need only specify that the entire {value} generated by the agent's effort will be transferred to the agent, and this is sufficient incentive for the agent to choose the welfare-maximizing effort.

Our starting point for this work is the observation that for multiple agents, this is no longer the case, as rational agents will attempt to free-ride on the efforts of their peers. Welfare maximization thus requires a coordination device---and contracts seem just the right tool to achieve this. However, the design of welfare-maximizing contracts has not been systematically addressed to date, possibly due to the triviality of maximizing welfare for a single agent (see the open question posed recently in \cite{dutting2024algorithmic}).
\emph{In this paper we aim to fill this gap, by studying contracts for (approximately) optimizing social welfare in multi-agent settings.} 

\paragraph{Setting.}

Our setting is the standard multi-agent contract setting~\cite{original-contract,multi-agents}, which we present in a slightly nonstandard way: There is a pool of $n$ agents, and a project with a \emph{value function} $f:2^{[n]}\to\mathbb{R_+}$, where $f(S)$ is the value of the project if the set of agents $S\subseteq [n]$ works on it as a team.%
\footnote{A more standard, equivalent interpretation of $f(S)$ (after normalization) is as the probability that the project succeeds if team $S$ works on it, generating a reward of $1$ for the principal if it does.} 
The value function $f$ is succinctly represented and accessed via value and/or demand oracles.
Each agent $i$ assigned to the project can choose between working and shirking, where the agent's cost for working is $c_i\ge 0$ and shirking costs nothing. The principal cannot directly observe each agent's choice, only the final value $f(S)$. This creates \emph{moral hazard}, the hallmark of contract design. 
To incentivize the agents to work, the principal designs a contract that allocates to each team member $i$ a share $\rho_i$ of the project's eventual value.%
\footnote{$\rho_i$ is also known as agent $i$'s linear contract. Linear contracts are without loss of generality in multi-agent settings.}
Allocating shares to incentivize effort is common practice, with applications ranging from start-up ventures to scientific research projects.

Given a team $S$, it is well-known how to design the share $\pay_i$ such that in equilibrium, agent $i\in S$ is incentivized to work: Let $\pay_i:=c_i/f(i: S)$, 
where $f(i:S)$ is the marginal contribution of $i$ to the value created by $S$.%
\footnote{Formally, for every $i\in[n]$ and $S\subseteq [n]$, $f(i:S):=f(S\cup i)-f(S\setminus i)$.}
That is, an agent's share is proportional to his cost and inverse-proportional to his marginal contribution. The latter may seem counterintuitive at first glance, but if an agent's work hardly makes a difference in terms of the project's value, his share of the value must be large for him not to shirk. In this case, it may not be beneficial to include the agent in $S$ to begin with.

Once $S$ is fixed, the utility of every agent $i\in S$ is $\pay_i f(S) - c_i$. The principal's utility $g(S)$ is $(1-\rho(S))f(S)$, where $\rho(S)=\sum_{i\in S} \pay_i$ is the total share of the project allocated to the agents. The welfare $w(S)$ is $f(S)-c(S)$, where $c(S)=\sum_{i\in S} c_i$ is the total cost of the agents in $S$ for exerting effort. Note that as usual in economics, the welfare is the sum of the players' utilities. 

\subsection{Research Questions} 

The main algorithmic challenge is in optimizing the team~$S$.
Importantly, not all sets of agents are admissible; a set must be \emph{feasible} (a.k.a.~\emph{budget balanced}), meaning that the sum of shares is bounded:
\begin{equation*}
    \rho(S)\le 1.
\end{equation*}
More generally, a set is said to be $b$-feasible if $\rho(S)\le b$.
This leads to an algorithmic question:

\begin{question}[Approximation]
\label{Q1}
Given a multi-agent contract setting with $n$ agents and value function~$f$, design a polynomial-time algorithm (with value and/or demand oracle access to $f$) that finds an approximately welfare-maximizing team $S$ among all feasible teams.
\end{question}

The objective of maximizing welfare subject to feasibility is well-motivated in scenarios where the principal's project is not for profit.
For example, in government contracts, maximizing the social welfare (rather than the principal's utility) is a primary concern (e.g.,~\cite{KangM22}). 
Cutting-edge work on AI agents explores the role of contracts in enhancing cooperation among the agents, with the goal of solving social dilemmas and increasing social good (e.g.,~\cite{ChristoffersenH23,Yocumetal23,ivanov2024principal}). 
The study of welfare-maximizing contracts is also an important complement to the intensive study of contract design for maximizing the principal's utility, which has been criticized as a framework that might be used by principals to exploit agents. 

When approaching Question~\ref{Q1}, we are in the interesting situation that contracts for maximizing the principal's utility in multi-agent settings have already been studied~\cite{original-contract,multi-agents,multi-multi,supermodular}.
In comparison, in auction design, maximizing welfare was understood much earlier than maximizing the auctioneer's utility (revenue)~\cite{Vickrey61,Clarke71,Groves73,Myerson81}.
This raises the question: Can we utilize the existing algorithms? This hinges on establishing a connection between the principal's utility and welfare, but in auction design, welfare is usually bounded away from revenue except for very special cases (like single-parameter bids drawn from monotone hazard rate distributions). Is this the case in contract design?

\begin{question}
\label{Q2}
In multi-agent contract settings, the principal's utility is trivially upper-bounded by welfare;
how large can the ratio between welfare and utility become, as the number of agents $n$ grows large and for different classes of value functions? 
\end{question}

We refer to the ratio in Question~\ref{Q2} as the \emph{welfare-utility gap}. Question~\ref{Q2} is mathematical/economic in nature (rather than algorithmic); the welfare-utility gap reflects how aligned the principal's interests are with those of a social planner. 
However, a constant gap would have algorithmic implications, since the approximation guarantees of existing algorithms for maximizing the principal's utility would expand to welfare with only a constant loss.

In \emph{single}-agent contract settings, the welfare-utility gap is fully characterized. It is known to coincide with the number of actions available to the agent~\cite{BalmacedaBCS16,DuttingRT20} (in our setting the available actions are work and shirk). For \emph{multiple} agents, the gap can potentially be much larger.
The best known upper bound for $n$ agents is $2n$, recently shown by \cite{multi-multi}. They compare the optimal principal's utility to $\max_{S\subseteq [n]} \{f(S)-c(S)\}$, which is the optimal \emph{unconstrained} welfare by a (possibly infeasible) team $S$.
They upper-bound this larger gap by the total number of actions summed across all agents, which is $2n$ in our multi-agent settings. 
The question of whether the gap for $n$ agents indeed grows with $n$ or not was open prior to our work. 

\subsection{Our Contribution}

We provide answers to Questions~\ref{Q1} and \ref{Q2} for a range of multi-agent contract settings.
Our investigation spans several classes of set functions to which the value function $f$ can belong, including the complement-free hierarchy of 
additive $\subsetneq$ submodular $\subsetneq$ XOS $\subsetneq$ subadditive (see, e.g.,~\cite{LehmannLN06}). We also consider supermodular value functions, which admit complementarities. 
Our main positive results are for XOS and submodular value functions, and are summarized in \cref{tab:tiny}. 

\input{table_tiny}

\paragraph{Results in \cref{tab:tiny} for general agents: Gaps.} 

In relation to Question~\ref{Q2}, our results include the first constant upper bounds on the welfare-utility gaps for submodular and XOS value functions. 
A constant gap means that the optimal principal's utility approximates the optimal welfare up to a constant factor. This is arguably surprising and has interesting implications: 
(i)~Any existing algorithm for finding a team (and the corresponding contract) such that the \emph{principal's utility} is approximately maximized, also approximately maximizes \emph{welfare}, with a loss of a constant multiplicative factor in the approximation guarantee. 
(ii)~In the other direction, a constant gap of~$r$ means that any lower bound for \emph{welfare} approximation exceeding $r$ implies a lower bound for approximating the \emph{principal's utility}. 

From a technical perspective, establishing a constant upper bound on the welfare-utility gap for XOS poses the following challenge: 
Consider two teams $T\subseteq S \subseteq [n]$ and an agent $i\in T$. What is the relation between $i$'s share $c_i/f(i: T)$ as part of team $T$, and $i$'s share $c_i/f(i: S)$ as part of the larger team $S$? 
If the value function $f$ is submodular, then $i$'s marginal contribution to the value created by $T$ is higher than to the value created by $S$, and so $i$'s share increases as part of the larger team.
This property is useful (e.g., it implies that feasibility of a team is downward-closed), and helps us establish an upper bound of $5$ on the welfare-utility gap for submodular settings (Proposition~\ref{submodular:UB}). 
For XOS, this property
is no longer guaranteed, complicating the analysis.
Nevertheless, we are able to establish a constant bound on the welfare-utility gap for XOS value functions in an \emph{algorithmic} way, as we now describe. 

\paragraph{Results in \cref{tab:tiny} for general agents: Approximations.} 

For XOS, we design a polynomial-time algorithm that computes a contract using value and demand queries, which guarantees the principal a constant fraction of the optimal welfare as 
their utility (see Theorem~\ref{1}). Thus, the optimal principal's utility is at least a constant fraction of the optimal welfare (see Corollary~\ref{XOS:demand:ratio}). Our algorithm's approximation ratio ($188$) is not small, but improves upon the previously best-known approximation guarantee for principal's utility, which was roughly $258$ (even though the previous guarantee holds with respect to a weaker benchmark). We anticipate that future work will further tighten the constants, making the theory more applicable.

For submodular functions, a $(1-1/e)$-approximate demand oracle can be simulated in polynomial time using only value queries~\cite{submodular-demand1, submodular-demand2}. 
Thus our algorithm for XOS also implies a polynomial-time $O(1)$-approximation algorithm for submodular value functions that uses only value queries, and also improves upon the previously known guarantee (that was roughly $642$) for principal's utility (see Corollary~\ref{XOS:submodular:ratio}). 

\paragraph{Lower bounds.} 

While they do not appear in \cref{tab:tiny}, we also establish lower bounds on welfare-utility gaps, summarized in \cref{symmetric_subadditive:prop:LB}. Our lower bounds show that XOS is the \enquote{limit} for a constant welfare-utility gap in the complement-free hierarchy of classes: The welfare-utility gap is lower-bounded by $\sqrt{n}-1$ for (even symmetric) subadditive value functions, and is unbounded for (even symmetric) supermodular value functions.%
\footnote{A value function $f$ is 
symmetric if $f(S)=f(|S|)$ for every $S\subseteq [n]$.}
Also, note that the smallest upper bound in \cref{tab:tiny} on any welfare-utility gap is $5$, holding for submodular value functions; we show a lower bound of $4$ on the welfare-utility gap for (even symmetric) \emph{additive} value functions.

\newcommand{\classP}{\textsf{P}}
\newcommand{\classNP}{\textsf{NP}}

We also establish a basic lower bound on the approximation ratio that can be achieved in polynomial time (with oracle access) under standard complexity assumptions. We show 
that for submodular value functions, no polynomial-time algorithm can approximate welfare within a factor of less than $\frac{e}{e-1}$ using only value queries, unless $\classP = \classNP$ (see \cref{e/e-1}).%
\footnote{To obtain this result we adapt a lower bound of \cite{inapproximability}.
Additional lower bounds from the literature on utility can be translated to welfare and value; see~\cite{supermodular, multi-agents, inapproximability}.}

\paragraph{Results in \cref{tab:tiny} for symmetric agents.} 

We initiate an investigation of \emph{symmetric agents}, where the value function is symmetric and the agents have a uniform cost $c$ for working.%
\footnote{Some of our results extend to the case where only the value function is symmetric but costs are non-uniform across agents, e.g., our results for symmetric XOS functions hold for this case.} 
This case is of practical interest as it captures central applications like crowdsourcing markets~\cite{crowdsource}, where crowdsourcing platforms can make it hard to distinguish between agents, and productivity is largely determined by the team's volume. 
It is also of theoretic interest, since we know from the algorithmic study of combinatorial \emph{auctions} that the understanding of symmetric settings (in particular, multi-unit rather than multi-item auctions) can shed light on the design problem~\cite{DobzinskiN10,nisan2015algorithmic}. 

We give an asymptotic upper bound of $16$ on the welfare-utility gap for symmetric settings with XOS value functions, which we refer to in short as \emph{sXOS} (Theorem~\ref{3}). As for approximation, the constant welfare-utility gap immediately implies that existing algorithms for approximating the principal's utility yield constant---but large---approximation ratios. Can the approximation guarantees be improved by a more direct design? 

To design an approximation algorithm for the symmetric case, we must revisit the question of how a symmetric instance is represented:
The representation size of symmetric, $n$-agent contract settings is of order $\log n$ rather than $n$.%
\footnote{E.g., we query $f(k)$ where $k\le n$ is the size of the team, rather than $f(S)$ where $S\subseteq [n]$, so we need $\log n$ rather than $n$ bits to describe the query.}
Thus, our algorithms for computing (approximately) optimal contracts for such settings (Theorem~\ref{4} and Proposition~\ref{s-submod:binary}) run in \enquote{logarithmic} time, just like the algorithms developed in the literature for multi-unit auctions (\enquote{logarithmic} is a slight misnomer as the running time is in polynomial in the representation size).
This is an advantage when the number of agents $n$ is very large, as is the case e.g.~in crowdsourcing.

For sXOS, our direct algorithm achieves a $(2+o(1))$-approximation using only value queries (Theorem~\ref{4}), which is tight (see \cref{2_is_tight}). For symmetric submodular, a simple binary search approach suffices to (exactly) maximize welfare (Proposition~\ref{s-submod:binary}).

\paragraph{Generalizations beyond welfare.}

In addition to welfare maximization, we extend our study of multi-agent contract design 
to maximizing the project's value $f(S)$ subject to feasibility constraints on $S$. Maximizing the value captures, e.g., an early-stage venture and public institutions aiming at maximizing output subject to financial constraints.%
\footnote{In economic terms this can be considered the \emph{revenue}, which becomes \emph{net income} after subtracting the agents' compensations.}
Unlike the welfare-utility gap, the value-welfare gap is unbounded even for a single agent (see \cref{Value-Walfare-Gap}). In Proposition~\ref{thm:XOS_value} we give our main result for this objective---a constant-factor approximation algorithm for XOS value functions using demand queries (and thus also for submodular value functions using only value queries). The approximation guarantee is a factor $4$ improvement relative to our welfare and utility approximation guarantee. 

Interestingly, this result immediately extends to approximating the value under any $b$-feasibility constraint, i.e., requiring that $\rho(S)=\sum_{i\in S}\frac{c_i}{f(i:S)}\le b$ where $b>0$. This is since, if we divide by $b$ both sides of the inequality constraint, we get the standard feasibility constraint (i.e., $\rho(S) \le 1$) for the scaled costs $c_i/b$, without altering the objective function $f$. Note that the same does not hold for the welfare and utility objective functions, where scaling the costs does alter the objective. In \cref{pro:XOS_transfer} we obtain a similar result for a constraint on the total \emph{transfer}: $\rho(S)f(S)\le B$ where $B>0$.%
\footnote{Essentially since such a constraint translates to $b$-feasibility for $b=\frac{B}{f(S)}$, and $f(S)$ for the optimal $S$ can be estimated.} 
In both cases, a small budget $b$ or $B$ (i.e., $b<1$ or $B<f(S)$, respectively) models a scenario in which the principal has limited flexibility in expending resources. A large budget ($b>1$ or $B>f(S)$, respectively) models a subsidized project. 

In \cref{subsection:value:corollaries}, we observe that for any $b$-feasible set $S$, since $\rho(S)\le b$, it holds that $w(S)\ge g(S)\ge (1-b)f(S)$.
Consequently, for $b=1-\Omega(1)$,
the value-welfare and value-utility gaps subject to $b$-feasibility are constant.
This holds for any class of value functions, and implies approximations for budgeted welfare and utility (see \cref{subsection:value:corollaries}).

Finally, for symmetric XOS value functions, for any $b=o(n)$ (even when $b>1$), the value-welfare gap diminishes as the team grows larger (provided it remains $b$-feasible).
This implies that our approximation results for welfare translate to value, for any class of value functions up to symmetric XOS, and thus generalize also to welfare and value subject to $b$-feasibility;
see \cref{value:symmatric}.

\paragraph{Paper organization.} 

We formally introduce the model in \cref{section:model}. 
In \cref{section:ratio} we focus on Question~\ref{Q2}, showing upper and lower bounds on the welfare-utility gap. The upper bounds are for general submodular value functions and for symmetric XOS, stopping short of general XOS value functions, which are deferred to \cref{section:XOS}. 
Our welfare-approximation algorithm for general XOS appears in \cref{section:XOS}, along with its immediate implications---including upper bounds of the welfare-utility gap for general XOS. In \cref{section:sublinear} we show our tight welfare-approximation algorithm for symmetric XOS.
\cref{section:value} includes an approximation algorithm for the alternative objective of value and its implications.
Some additional results and proofs are deferred to the appendices. 

\subsection{Related Work}

The theory of contracts is well-developed in economics~\cite{Nobel}. But modern contracts are increasingly applied in complex, computerized environments such as online labor markets, influencer marketing platforms, global afforestation programs, and sophisticated healthcare systems~\cite{crowdsource,socialmedia,deforestation2,LiIL21,medicare}. The growing complexity motivates the development of new computational approaches~\cite{original-contract,crowdsource,minmax1}. The algorithmic toolbox for handling complexity has already been successfully applied to several combinatorial aspects of modern contracts (e.g.,~\cite{DuttingRT20,DuettingEFK21}).

\paragraph{Relation to multi-agent contracts for utility.} 

In multi-agent settings, the fact that welfare maximization  has not yet been studied is intriguing given the body of knowledge on maximizing the principal's utility. The works of \cite{multi-agents,multi-multi,supermodular,inapproximability} give algorithms and lower bounds for approximating the utility-maximizing linear contract profile among all feasible profiles, in a range of settings accessed through value or demand queries. These include settings in which the project’s value function $f$ belongs to the complement-free hierarchy of set functions, as well as settings where it exhibits complementarities (supermodular $f$). 

To establish our constant welfare-utility gap for XOS in Section~\ref{section:XOS}, we repurpose two techniques developed in \cite{multi-agents} for utility. 
Recall that we design an algorithm (Algorithm~\ref{alg3}) that returns a team of agents, guaranteed to generate utility that is a constant approximation to the optimal welfare. First, our algorithm builds upon a useful and elegant scaling technique from \cite{multi-agents} (see Algorithm~\ref{alg1} included in Appendix~\ref{appendix:XOS} for completeness). Second, we generalize the approach in~\cite{multi-agents} to search for a suitable input to Algorithm~\ref{alg1}. Our generalization appears in Algorithm~\ref{alg2} and its stronger guarantees in Lemma~\ref{lemma:alg2}. 
Our Algorithm~\ref{alg3} guarantees an improved constant approximation factor to the optimal utility compared to the guarantee in~\cite{multi-agents}, and with respect to a stronger benchmark of the optimal welfare.
Interestingly, using our constant upper bound of $188$ on the welfare-utility gap for XOS, it immediately follows that the algorithm in \cite{multi-agents} designed for approximating optimal utility and guaranteeing an approximation factor of $258$, actually also approximates the optimal welfare up to a factor of $258\cdot 188 = 48,504$. Our Algorithm~\ref{alg3} guarantees a $188$-approximation to the optimal welfare. We do not know whether the analysis of either algorithm is tight. 

The work of \cite{feldman2025budget}, which has recently been made public, is closely related to ours; while many results overlap, their focus is on budgeted contracts, whereas we focus on welfare and value subject to feasibility (with some extensions to $b$-feasibility and $B$-transfer directly obtained from our analysis of welfare and value maximization). They also focus on a general reduction between objectives; while we study gaps as well as direct approximations, and obtain improved constant guarantees for welfare and value compared to currently known guarantees for utility, and tight guarantees in natural special cases. Like ours, the work of~\cite{feldman2025budget} focuses on XOS and submodular $f$, and they also consider $f$ that can be fed more than one action per agent (as in~\cite{multi-multi}). Our work lower bounds the gap beyond XOS, and also initiates a comprehensive study of symmetric $f$.
We view the combination of these two works as attesting to the timeliness of studying objectives beyond utility in contract design.

\paragraph{Additional works.} 

The works of \cite{CacciamaniEtAl24,hidden-contracts} study utility-maximizing multi-agent contracts, but under different multi-agent models in which the agents' performance is measured individually or in which the setting is represented explicitly.
Alon et al.~\cite{AlonLST23} study welfare-maximizing contracts, but in multi-\emph{principal} setting. While multiple principals pose entirely different challenges than multiple agents, their work nevertheless inspired our investigation of welfare in the context of multi-agent contracts.
The works of \cite{SaigTR23} and \cite{budget-contracts} design single-agent and multi-agent contracts, respectively, both subject to \emph{budget feasibility constraints}.
Their objectives are neither welfare nor principal's utility---rather, maximizing (a function of) the efforts---and their budget constraints are not tied to the project value. 

%% file: table_tiny.tex
\begin{table}[tp]
\centering
\begin{tabular}{|l||c|c|c|c|}
\hline
\multirow{2}{*}{\bf\makecell{~~Upper\\~~~Bounds}} & \multicolumn{2}{c|}{\textbf{General agents}} & \multicolumn{2}{c|}{\textbf{Symmetric agents}} \\
& \textbf{Approx} & \textbf{Gap} & \textbf{Approx} & \textbf{Gap} \\
\hline
\hline
\textbf{XOS $f$} & \makecell{$188$\\ Thm.~\ref{1}} & \makecell{$188$\\Cor.~\ref{XOS:demand:ratio}} & \makecell{$2+o(1)$\\ Thm.~\ref{4}} & \makecell{$16+o(1)$\\ Thm.~\ref{3}} \\
\hline
\textbf{Submodular $f$} & \makecell{$468$\\ Cor.~\ref{XOS:submodular:ratio}} & \makecell{$5$\\ Prop.~\ref{submodular:UB}} & \makecell{$1$ (OPT)\\ Prop.~\ref{s-submod:binary}} & \makecell{$5$\\Prop.~\ref{submodular:UB}} \\
\hline
\end{tabular}
\caption{\footnotesize Our main positive results for value functions that are XOS (upper row) or submodular (lower row). The agents can be either heterogeneous (left-hand side) or symmetric (right-hand side). Under ``Approx'' we specify the approximation ratios that our algorithms guarantee for welfare, answering Question~\ref{Q1}, and under ``Gap'' we specify our upper bounds on the welfare-utility gaps, answering Question~\ref{Q2}. Our approximation algorithms assume value oracle access to submodular functions, and both value and demand oracle access to XOS functions. A value query to a \emph{general} function requires specifying the set $S\subseteq [n]$, and our algorithms run in time polynomial in the query size~$n$. A value query to a \emph{symmetric} function requires specifying only the set's size $|S|\in [n]$, and our algorithms run in time polynomial in the query size $\log n$.} 
\label{tab:tiny}
\end{table}


%% file: CC_model.tex
\section{Model}
\label{section:model}

\paragraph{Instance.} 

We define a \emph{multi-agent  
contract setting} with a single principal and $n$ agents as a pair $(f, c)$, where $f : 2^{[n]} \to \mathbb{R}_+$ is the \emph{value} function, and $c : 2^{[n]}\to \mathbb{R}_+$ is the \emph{cost} function. 
Let $S \subseteq [n]$ denote the subset of agents who choose to work, then the value generated for the principal is $f(S)$. The value function is monotonically increasing, meaning that for every $S \subseteq T \subseteq [n]$, $f(S) \le f(T)$. 
We assume, as standard in the literature, that $f(\emptyset)=0$, and further assume that for any single agent $i$, its value is positive and upper bounds its cost.
Denote by $f(i:S):=f(S\cup \{i\})-f(S\setminus \{i\})$ the marginal contribution of agent $i$ to the value of a given subset $S$.
The cost function is additive, i.e., there is an individual cost $c_i := c(\{i\})$ for agent $i$, and $c(S)=\sum_{i\in S}{c_i}$. Individual costs are assumed to be strictly positive.

\paragraph{Interpretation.} 

The following is the standard interpretation of our setting in the literature (this interpretation is not necessary for the analysis and the reader may feel free to skip it): 
The project has two possible outcomes --- success and failure, where success has reward $1$ for the principal and failure has reward $0$. The subset $S$ of agents who exert effort determines the probability of success. The principal does not observe $S$, only the outcome of the project.
After normalization such that $f([n])= 1$, the value $f(S)$ can then be interpreted as the \emph{expected} reward of the project for the principal, when the subset of agents working on it is $S$. This interpretation clarifies that our definition above of a multi-agent contract setting is an equivalent (and arguably cleaner) way of describing standard multi-agent contract settings, and why it falls within the classic principal-agent model~\cite{holmstrom1979moral,grossman1992analysis}. In the classic model, distributions are an inherent part, since an agent chooses a hidden action that \emph{stochastically} leads to a rewarding outcome for the principal. We can describe the setting without using distributions since it suffices to specify the expected reward for each team.

\paragraph{Contract, utilities and welfare.}

Our focus is on \emph{linear} contracts, which is without loss of generality for multi-agent contract settings~\cite{dutting2024algorithmic}.
A (linear) contract $\rho$ maps every agent $i$ to $\rho(i)$, which is the fraction of the project's value (minus the base value $f(\emptyset$)) that will be allocated to agent $i$.%
\footnote{The standard (and equivalent) interpretation of $\rho(i)$, often denoted by $\alpha_i$ in the literature, is as the fraction of the principal's reward which is transferred to agent $i$.}    
In other words, agent $i$ is allocated a stake or \emph{share} $\rho(i)$ in the project. 
The principal keeps the remaining share $1-\sum_{i\in [n]}\rho(i)$ of the project. 

Given a contract $\rho$, let $S=S_\rho$ be the subset of agents incentivized to work. 
Agent $i$'s utility $\mu_i(S, \rho)$ is $\rho(i)f(S) - c_i$ if $i\in S$, and $\rho(i)f(S)$ otherwise, i.e., his fraction of the project's value, minus the cost of his effort if he is among the working agents. 
The principal's utility $\mu_p(S, \rho)$ is $(1 - \sum_{i\in [n]}\rho(i))f(S)$, i.e., the fraction of value retained after allocating the agents their share. 
The \emph{welfare} is
$$
w(S):=f(S) - c(S), 
$$
i.e., the project's value minus the total cost of work.
Notice that the welfare does not depend directly on the contract $\rho$---only through the team $S$ of agents incentivized to work---and that it is equal to the sum of all players' utilities: $w(S) = \mu_p(S,\rho) + \sum_{i\in[n]}\mu_i(S,\rho)$.

\paragraph{Incentivizing a team $S.$}
Since agents receive the same project share regardless of their (hidden) effort, each agent will opt to work only if the increase in their share's value as a result of their contribution outweighs their cost of exerting effort.
Fix a contract $\rho$ and consider the subset $S$ of agents with nonzero shares. This is the set of agents the contract wishes to incentivize to work (an agent $i\notin S$ will not work since they have positive cost but their share is zero).
We say that $\rho$ \emph{incentivizes} $S$ if every agent $i \in S$ weakly%
\footnote{We assume, as is standard in the literature, that an agent whom the principal wants to work and thus allocates a positive share to
will \emph{tie-break} in favor of the principal. That is, such an agent will exert effort if they receive nonnegative utility from working.
}
prefers to work: $\mu_i(S, \rho) \ge \mu_i(S \setminus {i}, \rho)$. 
This simplifies to $\rho(i)\ge c_i/f(i:S)$.
Let $\rho_S$ be the contract that incentivizes $S$ using the minimum share $\rho(i)=c_i / f(i:S)$ per agent.
We denote by $\rho(S)=\rho_S(S)$ the minimum total share needed to incentivize $S$.
We define the \emph{utility} function 
$$
g(S):=\mu_p(S,\rho_S)=(1-\rho(S))f(S)
$$ 
to be the principal's maximum utility under contract $\rho_S$, and use $g(i)$ for agent $i$'s utility under the same contract.

\paragraph{Feasibility.}

A contract is \emph{feasible} (a.k.a.~\emph{budget-balanced}) if $\sum_{i\in [n]}\rho(i) \le 1$, i.e., the agents' fractions sum up to at most $1$. We can generalize the feasibility constraint to $b$-feasibility by replacing $1$ with $b$ for some positive $b$. We say a set $S$ is $b$-feasible if $\rho(S)\le b$ (that is, if $\rho_S$ is $b$-feasible). 
A different feasibility constraint is a bound on the total transfer of value to the agents: We say a set $S$ satisfies the \emph{$B$-transfer} constraint for $B\ge 0$ if $\rho(S)f(S)\le B$. Equivalently, a (non-empty) set $S$ satisfies the $B$-transfer constraint if it is $b$-feasible for $b=\frac{B}{f(S)}$. 

\paragraph{Benchmarks and gaps.}

For any instance $(f,c)$ we denote by $OPT_{(f,c)}(w)$ the optimal welfare over all feasible sets given the multi-agent contract instance $(f, c)$. 
Similarly, let $OPT_{(f,c)}(g)$ and $OPT_{(f, c)}(f)$ denote the optimal principal's utility and the optimal value, respectively, over all feasible sets.
We omit $(f,c)$ from the notation where clear from context.

\begin{definition}[Welfare-utility gap]
\label{definition:ratio}
For any family $I$ of multi-agent contract instances, the \emph{welfare-utility gap} of $I$ is defined as 
$$
\sup\limits_{(f, c)\in I} \frac{OPT_{(f, c)}(w)}{OPT_{(f, c)}(g)}. 
$$
\end{definition}

\noindent The \emph{value-utility gap} of $I$ is similarly defined by replacing $OPT_{(f, c)}(w)$ with $OPT_{(f, c)}(f)$, and likewise we can define the \emph{value-welfare gap}. 

\paragraph{Symmetric agents.} 
We say a multi-agent contract setting is \emph{symmetric} if its value function is a symmetric set function that depends only on the number of working agents and not on their identities. Formally, for any two agent subsets $S,T$ such that $|S| = |T|$, it holds that $f(S) = f(T)$. 
If $f$ is symmetric we can use the simplified notation $f(k)$ for $k\in[n]$ (and similarly for $c(k)$ and $\rho(k)$ if the cost and/or contract are symmetric), and denote by $m:[n]\to \mathbb{R}_+$ the marginal contribution $m(k := f(k) - f(k - 1)$. 
We assume uniform costs $c(k)=c\cdot k$ for some $c>0$. In \emph{partially} symmetric settings, agent costs are heterogeneous, and we assume without loss of generality that the agents are ordered by increasing cost $c_1 \le ... \le c_n$, so that the first $k$ agents are always the best $k$-sized team. 

\paragraph{Value functions and oracle access.} 

Standard definitions of classes of value functions appear in \cref{Appendix: perliminaries} for completeness.
Access to a value function $f$ is via two types of oracles, \emph{value} and \emph{demand}.
A value oracle takes as input a set $S$ and returns its value $f(S)$. A demand oracle takes as input a cost vector $c$ and returns a \emph{demand set} $S$ that maximizes $w(S)=f(S)-c(S)$, 
and is not necessarily feasible. For an approximation parameter $a\in (0, 1]$, an $a$-demand oracle returns an $a$-demand set $S$ such that $w(S)\ge a\cdot f(S)-c(S)$ for every alternative set $T$ (see \cite{submodular-demand1, submodular-demand2, multi-agents}).

Consider a set of agents $A$. For every set function $f$ with domain $2^A$ (all subsets of $A$), for every subset $B\subseteq A$, we denote by $f_{|B}$ the \emph{domain restriction} of $f$ to $2^B$.  Similarly, a function $h$ is said to be a \emph{domain extension} of $f$, if $f=h_{|C}$ for some $A\subseteq C$.

%% file: FF_ratios.tex
\section{Welfare-Utility Gaps: sXOS, Submodular, and Lower Bounds}
\label{section:ratio}

In this section, we study the ratio between the optimal welfare and the optimal principal's utility. We give two upper bounds for the following classes of value functions: an upper bound of $16+o(1)$ for symmetric XOS (\emph{sXOS} for short; \cref{sub:ratio-ub-sxos}), and an upper bound of $5$ for general submodular (\cref{sub:ratio-ub-submod}). In \cref{sub:ratio-lbs} we establish three lower bounds, which apply even to symmetric functions: we show a lower bound of $4$ for the class of (symmetric) additive value functions, and demonstrate that the welfare-utility gap is no longer bounded by a constant beyond XOS by showing lower bounds for subadditive and supermodular functions.

\subsection[sXOS: An Upper Bound of 16+o(1)]{sXOS: An Upper Bound of $16+o(1)$}
\label{sub:ratio-ub-sxos}

\begin{theorem}\label{3}
    For any $n$-agent setting $(f, c)$ with an sXOS value function, the welfare-utility gap is upper bounded by $16+o(1)$.
\end{theorem}

\begin{remark}
    The upper bound in \cref{3} holds even for \emph{partially} symmetric XOS settings (with heterogeneous agent costs). 
    It is exactly $16$ when $n$ is a multiple of $4$.
\end{remark}

To prove \cref{3}, we first use two fundamental properties of sXOS functions to establish \cref{sXOS:lem:algorithm} below. 
Recall that for symmetric functions, $m:[n]\to\mathbb{R}_+$ is the marginal contribution function mapping $k$ to the marginal contribution of the $k$th agent $f(k) - f(k - 1)$.
The \emph{marginals property} in \cref{sXOS:marginals} states that 
$$
m(k)\le \frac{1}{k} \cdot f(k)
$$ 
for every $k\ge 1$. The \emph{concavity property} in \cref{sXOS:concave} states that for every $k\in[n]$ and parameter $a\in [0,1]$ such that $ak\in [n]$, it holds that 
$$
a \cdot f(k) \le f(a k).
$$

Also, recall that in the partially symmetric setting we assume that the agents are ordered w.r.t.\ their costs in non-decreasing order $c_1\le ...\le c_n$. 
Thus, we can restrict attention to sets of agents of the form $[k]$ such that $k\in [n]$ (i.e., the first $k$ agents, who have the lowest cumulative cost among all subsets of $k$ agents). We write $w(k)=f(k)-c(k)$, where $c(k)=c([k])$. 
To bound the maximum ratio between optimal principal utility and optimal welfare, we assume, without loss of generality, that the set that maximizes welfare is $S^* = [n]$.%
\footnote{Otherwise, we can restrict the instance to $f_{|S^*}$, where $S^*$ is the set that maximizes welfare. The optimal welfare stays the same, while the optimal utility can only get smaller by restricting the domain.} 

\begin{lemma}
\label{sXOS:lem:algorithm}
   Consider a procedure that gets as input arguments $b>a\ge 1$ where $a$ divides $n$, initializes $k\gets n/\aL$, and decreases $k$ as long as the marginal contributions satisfy $m(k)\le \frac{\aL}{\bL}\cdot m(n)$. Upon termination, for the value $k > 0$ of the last iteration the two following properties hold:
   \begin{enumerate}
       \item $f(k)> (\aL^{-1}-\bL^{-1})f(n)$;
       \item $g(k)> (\aL^{-1}-\bL^{-1})(1-\bL\cdot\aL^{-2})f(n)$.
   \end{enumerate}
\end{lemma}

Before we prove \cref{sXOS:lem:algorithm}, let us show how to use it to prove \cref{3}:

\begin{proof}[Proof of \cref{3}]
    From \cref{sXOS:lem:algorithm} and since $f(n)\ge w(n)$, it holds that
    \begin{align*}
        g(k) > (\aL^{-1}-\bL^{-1}) (1-b\cdot \aL^{-2}) w(n).
    \end{align*}
    Since $OPT(g)\ge g(k)$ and $OPT(w)=w(n)$, we see that (the inverse of the gap) $OPT(g)/OPT(w) > r(\aL, \bL)$, where $r(\aL, \bL):=(\aL^{-1}-\bL^{-1})(1-\bL\cdot \aL^{-2})$. To minimize the upper bound on the gap, we maximize the lower bound $r(\aL, \bL)$ on the inverse.
    Standard calculations show that for any $\aL\ge 1$, substituting $\bL=\aL^{3/2}$ maximizes the expression.
    Using Lagrange multipliers, we find that the maximum is reached when $\aL=4$ and $\bL=8$. For these parameters, we have that $r(a,b)\le \frac{1}{16}$. 
    The upper bound of $16$ is exact for $\aL=4$ for any $n$ divisible by $4$.  For $n\le 16$, the bound of 16 trivially holds since the singleton agent is an $n$-approximation, an observation that we discuss in \cref{appendix:subadditive:ub} for subadditive functions. For all other values of $n$, the bound becomes $16+o(1)$ since
    in this case we can choose the parameter $a$ for which $\frac{n}{a}=\lfloor \frac{n}{4}\rfloor$, and, since $n>16$, this value is in the range between $4$ and $4+o(1)$.
\end{proof}

It remains to prove the lemma:

\begin{proof}[Proof of \cref{sXOS:lem:algorithm}]
    We first show that the procedure indeed stops at some $k>0$. Otherwise, we get the following contradiction:
        \begin{alignat*}{1}
        f(n/\aL) \; & = \; f(n/\aL)-f(0) \; = \; m(n/\aL)+...+m(1)
        \\&\le \; \frac{n}{\aL}\cdot\frac{\aL}{\bL}\cdot m(n) \; = \; \frac{n}{\bL} \cdot m(n)
        \\&\le \; \frac{n}{\bL}\cdot\frac{f(n)}{n}
         \; = \; \frac{1}{\bL} f(n)
        \\& <  \; \frac{1}{\aL} f(n)
         \; \le \; f(n/\aL).
    \end{alignat*}
    The first inequality follows since the algorithm is assumed towards contradiction to stop at $k=0$, the second inequality follows from the {marginals property}, the third inequality is by definition since $b>a$, and the fourth inequality is from the {concavity property}.
        
    Turning to property (1) of the lemma, since $k>0$ we have that
    \begin{align*}
        f(n/\aL)-f(k)&= m(n/\aL) + ... + m(k+1)
        \\& \le \left(\frac{n}{a}-k\right)\cdot \frac{a}{b} \cdot m(n)
        \\&<\frac{n}{a}\cdot\frac{a}{b}\cdot m(n)\; = \; \frac{1}{b}f(n).
    \end{align*}
    The first inequality is similar to the contradiction case: by the definition of the algorithm $m(k)\le \frac{a}{b}m(n)$ and we have ($\frac{n}{a} - k$) such marginals that we add.
    The second inequality holds since the algorithm stops at $k$.
    Thus, $ f(k)> f(n/\aL)-\frac{1}{\bL}f(n) \ge \frac{1}{\aL}f(n)-\frac{1}{\bL}f(n)$ where the second inequality is by the concavity property.
    
    Towards establishing property (2) of the lemma:
    \begin{align*}
        g(k)&= f(k)\left(1-\frac{c(k)}{m(k)}\right)
        \\&\ge f(k)\left(1-\frac{c(n)}{\aL}\cdot\frac{1}{m(n)}\right)
        \\&\ge f(k)\left(1-\frac{c(n)}{\aL}\cdot\frac{\bL}{\aL}\cdot\frac{1}{m(n)}\right)
        \\&  = f(k)\left(1-\frac{\bL}{\aL^2}\cdot\frac{c(n)}{m(n)}\right)
        \\&\ge f(k)\left(1-\frac{\bL}{\aL^2}\right)
        \\&> (\aL^{-1}-\bL^{-1})f(n)(1-\bL\cdot\aL^{-2}),
    \end{align*}
where the first inequality follows since agent costs are ordered and since $k\le n/\aL$, the second inequality is due to the marginals property, the third inequality holds because $[n]$ is feasible, and the last inequality follows from property (1) shown above.
\end{proof}

\subsection[Submodular: An Upper Bound of 5]{Submodular: An Upper Bound of $5$}
\label{sub:ratio-ub-submod}

\begin{proposition} \label{submodular:UB}
    For any $n$-agent setting $(f, c)$ with a submodular value function, the welfare-utility gap is upper bounded by $5$.
\end{proposition}

To prove \cref{submodular:UB}, we use the following lemma.

\begin{lemma}\label{splitting_lemma}
    For any $n$-agent setting $(f, c)$ with a submodular value function, for every feasible set $S$ satisfying $\rho_S(i)\le \frac{1}{2}$ for every $i\in S$, there exists a partition of $S$ into three disjoint sets $A$, $B$ and $C$, such that $|B|=1$, and for each set $T\in \{A, B, C\}$ it holds that $\rho(T)\le \frac{1}{2}$.
    
\end{lemma}
\begin{proof}
    Consider the following procedure: Add the agents of $S$ to the set $A$ one by one as long as $\rho_S(A)\le\frac{1}{2}$. Then, add the next agent to the set $B$, and put the rest in $C$.
    We claim that the resulting three sets satisfy the requirements of the lemma. 
    
    First, clearly these sets represent a partition of $S$. Second, by construction, $\rho_S(A)\le \frac{1}{2}$, and by assumption $\rho_S(B)\le \frac{1}{2}$. 
    To show that $\rho_S(C)\le \frac{1}{2}$, notice that from additivity of $\rho_S$ it holds that $\rho_S(S)=\rho_S(A)+\rho_S(B)+\rho_S(C)$. Since $S$ is feasible it holds that $\rho_S(S)\le 1$.
    Assuming $C$ is not empty, $B$ must be non-empty. This means that $\rho_S(A\cup B)=\rho_S(A)+\rho_S(B)>\frac{1}{2}$, since otherwise the procedure would have added the agent from $B$ to the set $A$ instead, and thus $\rho_S(C)<\frac{1}{2}$.
    Finally, by submodularity, for every set $T\subseteq S$ it holds that $\rho_S(T)\ge \rho_T(T)$, since $f(i:S)\le f(i:T)$ for every $i\in T$. This concludes the proof.
\end{proof}

\begin{proof}[Proof of \cref{submodular:UB}]
    Consider the set $S$ maximizing welfare. We have two cases: either there is a single agent $i\in S$ for which $\rho_S(i)> \frac{1}{2}$, or for all $i\in S$ it holds that $\rho_S(i)\le \frac{1}{2}$.
    
    \paragraph{Case 1.}
    In this first case, since $1\ge \rho_S(S)=\rho_S(S\setminus i)+\rho_S(i)$ and since $\rho_{S\setminus i}(S\setminus i)\le \rho_{S}(S\setminus i)$ we have $\rho_{S\setminus i}(S\setminus i)\le \frac{1}{2}$. 
    Note that for any $S$ with $\rho_S(S) \le \frac 12$ we have
    $g(S)=(1-\rho(S))\cdot f(S) \ge \frac{1}{2}\cdot w(S)$ and thus, 
    \begin{equation}
        \label{eq:2bound}
        w(S) \le 2 g(S).
    \end{equation}
    This implies
        $$ w(S) 
        \; \le \; w(i) + w(S\setminus i)
        \; = \; g(i) + w(S\setminus i)
        \; \le \; g(i) + 2g(S\setminus i)
        \; \le \; 3\;{OPT}(g),
        $$
    where the first inequality holds by the subadditivity of $w$ inherited from $f$, the equality since $w(i)=g(i)$, and the second inequality follows from Equation~\eqref{eq:2bound}.

    \paragraph{Case 2.}
    In the second case, we use \cref{splitting_lemma} and consider the partition of $S$ into three disjoint sets $A$, $B$ and $C$, where $|B|=1$ and for each set the total share is upper bounded by $\frac{1}{2}$. Thus, 
    $$ 
    w(S) \; \le \; w(A)+w(B)+w(C) \; \le \; 2g(A) + g(B) + 2g(C) \le \; 5\; OPT(g),
    $$
    where the first inequality is by subadditivity of $w$, and the second is from Equation~\eqref{eq:2bound} and since $B$ is a singleton.
\end{proof}

\subsection{Lower Bounds}
\label{sub:ratio-lbs}
\label{additive:LB:4}
\label{supermodular:ratio}

We present three lower bounds on the welfare-utility gap, for subadditive, additive and supermodular value functions. They apply even in the symmetric setting with uniform costs.

\begin{proposition}
\label{symmetric_subadditive:prop:LB}
$ $
\begin{enumerate}
    \item For every $n\ge 5$, there exists an $n$-agent setting with a symmetric subadditive value function and uniform costs, for which the ratio between optimal welfare and optimal utility is at least $\sqrt n - 1$.
    
    \item There exists an $n$-agent setting with a symmetric additive value function and uniform costs for which the ratio between optimal welfare and optimal utility is at least $4-\frac{4}{n}$.

    \item For every $n\ge 2$, there exists an $n$-agent setting with a symmetric supermodular value function and uniform costs, for which the ratio between optimal welfare and optimal utility is unbounded.

\end{enumerate}
\end{proposition}

\begin{proof}
We prove each of the three lower bounds stated in the proposition. 

\paragraph{Subadditive lower bound.}

Let $n\ge 5$. We begin by defining a \emph{convex interpolation}: We say that a function $A:\{0,\ldots,n-1\}\to [a, b]$ is a convex interpolation if $A(k)=\left(1-\frac{k}{n-1}\right) a + \frac{k}{n-1}b$. We will consider $A:\{0,\ldots,n-1\}\to[\sqrt{n}, n-\sqrt{n}]$, which becomes
    \begin{align*}
        &A(k)=\left(1-\frac{k}{n-1}\right)\sqrt{n}+\frac{k}{n-1}\cdot (n-\sqrt{n})
        =\sqrt{n}+k\cdot \frac{n-2\sqrt{n}}{n-1}.
    \end{align*}
    
    Let us define a value function $f$ as follows: $f(0):=0$, $f(n):=n$, and for every other $k$, $f(k):=A(k)$.
    For the cost we set $c =  \frac{1}{\sqrt{n}}$.
    
    The function $f$ is subadditive, as for any $m\in [n]$ and any $k\in [n-1]$,
    $f(k) + f(m-k)=2\sqrt{n}+m\frac{n-2\sqrt{n}}{n-1}>  f(m)$. We observe that $f$ is monotone increasing: Its marginals are $m(1) = m(n)=\sqrt{n}$, and for every other $k$, $m(k) = \frac{n-2\sqrt{n}}{n-1}$ is positive since $n\ge 5$.
    
    Notice that the sets $[1]$ and $[n]$ are both feasible as $\rho(1)=\frac{c}{m(1)} = \frac{1}{n}$
    and $\rho(n)=n\cdot \frac{c}{m(n)} = 1$.
    We claim that the optimal utility is achieved by set $[1]$. To see this, observe that $g(n)=0$ and for any $k\in [n-1]$ larger than $1$:
    \begin{align*}
        g(k) &= \left(1-\frac{k}{\sqrt{n}} \cdot \frac{n-1}{n - 2\sqrt{n}}\right) \cdot \left(\sqrt{n} + k\cdot\frac{n-2\sqrt{n}}{n-1} \right) \\
        &< \left(1-\frac{k}{\sqrt{n}}\right) \cdot \left(\sqrt{n} + k \right) \\
        &< \sqrt{n} - \frac{k^2}{\sqrt{n}} 
        \; \le  \; \sqrt{n} - \frac{1}{\sqrt{n}} 
        \\ &= \left(1-\frac{1}{n}\right)\sqrt{n}\;  = \; g(1).
    \end{align*}
    Thus, the welfare-utility gap is lower bounded by 
    $$w(n)/g(1) = (n-\sqrt{n})/(\sqrt{n} - 1/\sqrt{n}) = (\sqrt{n} - 1)/(1-1/n) \ge \sqrt{n}-1.$$

\paragraph{Additive lower bound.}
    Let $f$ be such that for every $k$, $f(k)=k\cdot n$, and let $c=1$.
    The marginals are $m(k)=n$ and the share is $\rho(k)=\frac{k}{n}$.
    Note that the welfare function, $w(k)=k\cdot n-k$, is increasing for any $k$, and that $[n]$ is feasible as $\rho(n)=1$. Thus, optimal welfare is reached when $k=n$ and $OPT(w)=w(n)=n^2-n$.
    Towards optimal utility, note that $g(k)=(1-\frac{k}{n})\cdot k\cdot n$ for any $k$ which is maximized by $k=\frac{n}{2}$. Therefore, $OPT(g)=g(n/2)=\frac{n^2}{4}$. The ratio between $OPT(w)$ and $OPT(g)$ is $4\cdot \frac{n^2-n}{n^2}=4-\frac{4}{n}$, which approaches $4$ when $n\to \infty$.

\paragraph{Supermodular lower bound.}
    Consider the function $f:[n]\to\mathbb R_+$ with $f(k)=k^2+k-\frac{k}{n}$ for $k<n$ and $f(n)=n^2+n$. Note that the marginals are increasing for every $k$, so $f$ is supermodular. For the uniform cost, we set $c=2$. Now for any $0<k<n$, the set $[k]$ is infeasible. The marginals are $m(k)=2k-\frac{1}{n}$, and thus $\rho(k)=\frac{2k}{2k-\frac{1}{n}}>1$. In contrast, $m(n)=2k$, so $\rho(n)=1$ and $[n]$ is feasible. This shows that $OPT(g) = g(n) = 0$, and $OPT(w)=w(n)>0$.
\end{proof}

%% file: DA_alg3.tex
\section{Welfare Approximation and Welfare-Utility Gap for XOS}
\label{section:XOS}

In this section we prove \cref{1}, which asserts that for XOS value functions, there exists an $O(1)$-approximation polynomial-time algorithm for optimal welfare. Furthermore, this algorithm returns a set whose \emph{utility} approximates the optimal welfare. 

\begin{theorem} \label{1}
    There exists an algorithm that outputs a set whose utility is an $O(1)$-approximation 
    of the optimal welfare, for any $n$-agent setting with an XOS value function $f$. The algorithm runs in time polynomial in $n$ using value and $a$-demand oracles.
\end{theorem}

\cref{1} implies two corollaries.

\begin{corollary}\label{XOS:demand:ratio}
  For any $n$-agent setting with an XOS value function, the welfare-utility gap is a constant.
\end{corollary}

Moreover, for submodular value functions, an $a$-demand oracle with $a=(1-\frac{1}{e})$ can be simulated efficiently using only value queries \cite{submodular-demand1, submodular-demand2}. The following corollary can be obtained. 

\begin{corollary}\label{XOS:submodular:ratio}
    There exists an algorithm that outputs a set whose utility is an $O(1)$-approximation of the optimal welfare, for any $n$-agent setting with a submodular value function. The algorithm runs in time polynomial in $n$ using value oracle.
\end{corollary}

See the approximation factors our algorithm obtains in \cref{approximation_factors}.
The rest of the section is organized as follows. In \cref{sub:XOS-approx-alg} we describe and analyze our algorithm (\cref{alg3}), which establishes \cref{1} and \cref{XOS:submodular:ratio} (\Cref{XOS:demand:ratio} follows immediately from \cref{1}). The analysis uses the guarantees of \cref{alg3} as summarized in \cref{fundamental}, which we prove in \cref{sub:Alg3-lemma}. In \cref{sub:XOS-stronger-benchmark} we discuss extensions of our results to $b$-feasibility.

\subsection{Approximation Algorithm}
\label{sub:XOS-approx-alg}

Our approximation algorithm (\cref{alg3}) uses two other algorithms. The first is \cref{alg1} from \cite{multi-agents}, 
which for any XOS function $f$ (with value and $a$-demand oracle access) and target value $0 \le y \le f([n])$, finds a set $U$ such that (1)~for all $i \in U$, the marginal $f(i: U)$ is approximately $f(i:[n])$; and (2)~$f(U)$ is close to $y$. We say that $(y,f,c)$ is \emph{a valid input} for this algorithm if $0 \le y \le f([n])$. Details appear in \cite{multi-agents} and, for completeness, in \cref{appendix:alg1}.

\paragraph{Algorithm~\ref{alg2}.} 

The second algorithm is a generalized form of an algorithm introduced by~\cite{multi-agents} (also named Algorithm~2). 
While \cref{alg1} encapsulates basic properties of XOS functions, \cref{alg2} connects these to near-optimal contract design. 
It attempts to construct a valid input for \cref{alg1} by scaling costs, restricting domains, and setting an appropriate target value. 
Our \cref{alg2} appears in \cref{appx:alg2-and-proof} and uses several new ideas, e.g., we use a property we call $\kappa$-boundedness, which concerns domain extensions of domain-restricted functions. 

The following lemma, which we prove in \cref{appendix:XOS}, captures the guarantees of \cref{alg2}.
Let us introduce some notation:
Let $x$ be a function that returns the maximum single-agent value of $f$; that is,
$x(f):=\max_{i\in [n]}f(\{i\})$. 
For a given value $a\in (0, 1]$ we define $A:=(a/4)^2$; subsequently, $1/A$ will be the loss from the use of an $a$-demand oracle. 
Let $M(f,y) := \frac{A}{m}\cdot \frac{y}{x(f)}$, where $m \ge 1$ is part of the input to \cref{alg2}. We also generalize our notation above, and denote by $\rho_T^{(h)}$ the function mapping a set to its total share while minimizing the total share needed to incentivize the set $T$, assuming the value function is $h$.

\begin{lemma}\label{lemma:alg2}
    Consider a $n$-agent setting $(f,c)$ with an XOS value function, a target value $y\ge 0$. 
    There exists a polynomial-time algorithm with $a$-demand oracle access to $f$ that guarantees the following. Given an input $(f, c, y, m)$ that satisfies the following validity conditions:
    
    \begin{enumerate}[nosep,noitemsep]
        \item There exists a set $S$ such that $y\le f(S)$; 
        \item There exists an XOS domain extension $h$ of $f_{|S}$ to domain $T$, s.t. $\rho^{(h)}_T(S)\le \left(2-\frac{1}{8m}\right)^2$; 
        \item $M(f, y)\ge 1$;
    \end{enumerate}
    the algorithm returns a set $U \subseteq D$ that satisfies
    \begin{enumerate}[nosep,noitemsep]
        \item $\frac{A}{m}\cdot y\ge f(U)\ge \left(1-\frac{1}{M(f, y)}\right)\cdot \frac{A}{2m}\cdot y$;
        \item $\rho(U)\le \frac{1}{m}$.
    \end{enumerate}
\end{lemma}

To interpret \cref{lemma:alg2}, think of $M(f, y)$ as sufficiently large. Then,  
Property (1) of the output implies that $f(U)$ approximates $y$. 
Property (2) of the output implies that $U$ is feasible. 
Since $g(U)=(1-\rho(U))f(U)$, for any constant $m>1$ it holds that $g(U)$ also approximates $y$. Moreover, since $w(U)\ge g(U)$, it holds that if $y$ is an estimate of the optimal welfare value, then so are $w(U)$ and $g(U)$.
Finally, notice that the ratio between $f(U)$ and $y$ depends on $M(f, y)=\frac{A}{m}\cdot \frac{y}{x(f)}$. 
Thus, for any target value $y$, assuming $A$ and $m$ are given constants, the ratio depends on the size of $x(f)$ relative to $y$, where $x(f)$ is determined by the agents in the domain $D$ of $f$.
Ideally, the maximum value $x(f)$ should be small compared to $y$. 
The key idea is to use \cref{lemma:alg2} with a restricted domain $D\subseteq [n]$ to reduce $x(f)$ to $x(f_{|D})$, then consider the original function $f$ as the domain extension $h$ of $f_{|S}$.%
\footnote{We remark that the $a$-demand oracle for $f$ extends to any of its domain restrictions, since we are allowed to put $p_i=\infty$ for any $i\notin D$ where $D$ is the domain restriction of $[n]$.}

\paragraph{Algorithm~\ref{alg3}.} 

Our main routine, \cref{alg3}, repeatedly invokes \cref{alg2} to ensure a good approximation for welfare. 
To obtain a good approximation, it must ensure three main conditions for \cref{alg2}: (1) $y$ is a target value close to the optimal welfare, (2) the largest agent value of $f$ is comparatively small, and (3) the parameters yield a valid input for \cref{alg1}. 
\cref{alg3} is designed to address these three issues.

First, since we do not know the value of the optimal welfare set in advance, \cref{alg3} systematically tests estimates $y$ within a given range by performing a doubling search. 

Second, to ensure that the ratio between $x(f)$ and $y$ remains small, \cref{alg3}  applies \cref{alg2} with a domain-restricted function $f_{|D_y}$ where $D_y$ excludes agents with a large value relative to $y$. 
Specifically, an estimate $M$ is chosen in advance, and for every $y$ the domain of $f$ is restricted to $D_y=\{i: f(i)\le \frac{A}{m}\cdot \frac{y}{M}\}$, which ensures that $M(f_{|D_y}, y)=\frac{A}{m}\cdot \frac{y}{x(f_{|D_y})}$ is at least as large as $M$. 

The third issue is that the reduced domain $D_y$ might not include the optimal welfare set $S$ as per the validity requirements of \cref{lemma:alg2}. \cref{alg3} ignores this issue and uses parameters that \emph{a priori} do not guarantee a valid input, and this issue will be dealt with in the analysis.

\setcounter{algocf}{2}
\begin{algorithm}[ht]
\DontPrintSemicolon
\caption{Constant-factor approximation for XOS value functions}
\label{alg3}

\KwInput{An instance $(f,c)$ where $f$ is a XOS function accessible through a value oracle and an $a$-demand oracle,
a doubling parameter $\gamma\in (0, 1)$, and parameters \gcm{$m\ge 1$}, $M>1$.}
\KwOutput{A collection of sets $\mathcal C$.}
\medskip
Let ${\mathcal C} \gets \left\{ \{i\} : i \in [n]\right\}$ and $y \gets \max_{i\in [n]}w(i)$.\;

\While{$y \le f([n])$}
{
Let $D_y\gets \left\{i\in [n]: f(\{i\})\le \frac{A}{m}\cdot \frac{y}{M}\right\}$. 

Run \cref{alg2} on $(f_{|D_y},\:c_{|D_y},\:y,\;m)$; add output to ${\mathcal C}$.

$y \gets y \cdot 2^\gamma$.
}
\Return ${\mathcal C}$

\end{algorithm}

We characterize the output of \cref{alg3} in the following lemma, which we prove in \Cref{sub:Alg3-lemma}:
\begin{lemma}
\label{fundamental}
     Let $\mathcal C$ be the collection returned by \cref{alg3}.
     For any feasible set $S$, the collection $\mathcal C$ contains at least one of the following sets:
    \begin{enumerate}[nosep,noitemsep]
        \item A set $U\subseteq [n]$ such that $\rho(U)\le \frac{1}{m}$ and $f(U)$ is a constant-factor approximation of $w(S)$.
        \item A singleton consisting of agent $i\in S$ 
        such that $g(\{i\})$ is a constant-factor approximation of $w(S)$.
    \end{enumerate}
\end{lemma}

The constant-factor approximation guaranteed by \cref{fundamental} is determined by the parameters $\gamma,m,M$ fed into \cref{alg3}, by the quality of the $a$-demand oracle it is provided with, and by a parameter $\delta\in (0, 1)$ that is used for case analysis in the proof. 
In more detail, we say that an agent is \emph{$\delta$-cheap} if its share when working alone is smaller than $(1-\delta)$, and otherwise, we say it is $\delta$-expensive. 
Determining the cutoff $\delta$ between these agent types, together with the additional four parameters $\gamma,m,M,a$, results in the approximation guarantee of \cref{fundamental} (see \cref{sub:Alg3-lemma} for details). 
In \cref{approximation_factors} we give two instantiations of these parameters, one for an exact demand oracle ($a=1$) and one for an approximate one ($a=1-\frac{1}{e}$), to derive \cref{1} and \Cref{XOS:submodular:ratio}, respectively.

\begin{proof}[Proof of \cref{1}]
        \cref{1} follows immediately from \cref{fundamental} since we are able to obtain a polynomial-sized collection of sets which is guaranteed to have at least one ``good" set. Specifically, consider the optimal welfare set $S$. Since $S$ is feasible, in both cases of \cref{fundamental}, we are guaranteed that the collection $\mathcal C$ contains a set $U$ whose utility is a constant approximation of $w(S)$, provided that $m$ is bounded away from $1$ (an even stronger guarantee holds for the first case, since $g(U)=(1-\rho(U))f(U)$, and so we get an approximation to the value---see \cref{section:value}).
        Therefore, the largest utility of a set $T\in \mathcal C$ is guaranteed to be a constant-factor approximation to the optimal welfare. 
        Hence, to obtain our desired approximation we output a set in $\arg \max_{T\in \mathcal C} g(T)$. 
\end{proof}

\subsection[Proof of Lemma \ref{fundamental}]{Proof of \cref{fundamental}}
\label{sub:Alg3-lemma}

\paragraph{Overview.}

    Consider a feasible set $S$, and consider the estimate $y$ of $S$; that is, the value $y$ which is in the range $f(S)\cdot 2^{-\gamma}<y\le f(S)$. 
    For simplicity, assume that $y$ is always precise, i.e., $y=f(S)$.
    We are interested in finding a set whose utility approximates $f(S)$.
    Observe that if the value $f(i)$ of every agent $i\in S$ is sufficiently small compared to $y$, then $S$ is contained in $D_y$. 
    Since $S$ is feasible, the conditions of \cref{lemma:alg2} hold with $h:=f_{|S}$. 
    Therefore, with input $f_{|D_y}$ \cref{alg2} will return a set $U$, whose utility $g(U)$ provides the desired constant-factor approximation of $y$.
    Otherwise, $S$ must contain an agent $i$ with a large value $f(i)$ relative to $y$. If the share of this agent is small, then agent $i$ alone provides the desired approximation of $y$ (since $g(i)=\left(1-\rho(i)\right)f(i)$). 
    
    Recall that by definition of $\rho$, for any set $S$, and any two agents $i, j\in S$, it holds that $\rho(i)+\rho(j)\le \rho(S)$. Thus, since $\rho(S)$ is bounded by 1, if $\rho(i)$ is large, then $\rho(j)$ must be small.
    This implies that for the set $S_i=S\setminus \{i\}$, and for the corresponding estimator $y_i$ of $f(S_i)$, one of the following must hold: (1) $S_i$ contains only individuals with small value relative to $y_i$, in which case $S_i\subseteq D_{y_i}$; (2) there exists an agent $j\in S_i$ with a large value and a small share. In both cases, we obtain a constant-factor approximation of $f(S_i)$.

    A natural way to proceed would have been to show that a set whose utility is a constant-factor approximation of $f(S_i)$ is also a constant-factor approximation of $f(S)$. However, this is not necessarily true: Consider the case where $S$ is a singleton. Then $S_i=\emptyset$, and thus $f(S_i)=0$.
    We show a weaker guarantee: either $f(S_i)$ is a constant-factor approximation of $w(S)$, or $g(i)$ is a constant-factor approximation of $w(S)$. This follows by subadditivity of $w$, together with the fact that $f(S_i)\ge w(S_i)$ and $g(i)=w(i)$.
    Either way, we obtain a set whose utility is a constant-factor approximation of $w(S)$.

\begin{proof}[Proof of \cref{fundamental}]
    Consider a feasible set $S$. 
    By construction, \cref{alg3} tests an estimate of this set's value up to a factor of $2^{-\gamma}$. 
    We say $y$ with $f(S)\cdot2^{-\gamma} < y \le f(S)$ is the \emph{estimate of $S$}.
    
    Recall that $D_y=\left\{i: f(i)\le \frac{A}{m}\cdot \frac{y}{M}\right\}$. We think of $D_y$ as the set of \emph{light} agents, where an agent is light (relative to $y$) if $f(i)\le \frac{A}{m}\cdot \frac{y}{M}$.
    Observe that if a set $S$ is composed of light agents and is a subset of a feasible set, then by \cref{lemma:alg2}, for an estimate $y$ of $S$, we are guaranteed to obtain
    a set $U$ with share $\rho(U)\le \frac{1}{m}$ and value $f(U)\ge \left(1-\frac{1}{M}\right)\cdot \frac{A}{2m}\cdot y$, and hence
    $$g(S)\ge \left( 1-\frac{1}{m} \right) \left( 1-\frac{1}{M} \right) \frac{A}{2m} \cdot  f(S) \cdot 2^{-\gamma}.$$
    
    If $S$ is not a subset of $D_y$, then it must contain a \emph{heavy} agent; that is, there is an agent $i\in S$ such that $f(i)> \frac{A}{m}\cdot \frac{y}{M}$. 
    Observe that a heavy agent's utility approximates the value of the entire set if and only if $\rho(i)$ is smaller than $1$ by a constant. This is since $g(i)=\left(1-\rho(i)\right)f(i)$. 
    
    Recall that we say an agent is \emph{$\delta$-cheap} if its share is smaller than $(1-\delta)$, and otherwise, we say it is $\delta$-expensive. 
    For a given $\delta\in (0, 1)$, if $i$ is $\delta$-cheap, then since $y\ge f(S)\cdot 2^{-\gamma}$ it holds that $$g(i)\ge \delta\cdot  f(i)>\frac{A}{m}\cdot \frac{\delta}{M}\cdot f(S)\cdot 2^{-\gamma}.$$
    
    If $i$ is \emph{$\delta$-expensive}, then since $\rho(S)\le 1$, every other agent in $S$ is $(1-\delta)$-cheap. Let us consider the set $S_i=(S\setminus i)$, and let $y_i$ be the estimator of $f(S_i)$.
    Since all agents in $S_i$ are $(1-\delta)$-cheap, there must be a set whose utility approximates $f(S_i)$: If every agent $j\in S_i$ is light (relative to $y_i$), \cref{lemma:alg2} guarantees a set $U'$ with $\rho(U')\le \frac{1}{m}$ for which
    $$g(U')\ge  \left( 1-\frac{1}{m} \right) \left( 1-\frac{1}{M} \right) \frac{A}{2m}  \cdot f(S_i)\cdot  2^{-\gamma}.$$
    In the other case, there exists a heavy agent $j\in S_i$. Then, since $\rho(j)\le \delta$ it holds that

    $$g(j)\ge \frac{A}{m}\cdot \frac{1-\delta}{M}\cdot f(S_i)\cdot 2^{-\gamma}.$$

Observe that in the last two cases we obtained an approximation of $f(S_i)$, while we aim for an approximation of $w(S)$.
By the subadditivity of $w$, either $w(S_i)$ or $w(i)$ is a constant-factor approximation of $w(S)$. Thus, either $g(i)$ (which is equal to $w(i)$) is larger than $\frac{A}{m}\cdot \frac{\delta}{M}\cdot w(S)$, or 
$$f(S_i)\ge w(S_i)\ge w(S)-w(i)\ge \left(1-\frac{A}{m}\cdot \frac{\delta}{M}\cdot 2^{-\gamma}\right)w(S)\ge \left(1-\frac{A}{m}\cdot \frac{\delta}{M}\right)w(S)\cdot 2^{-\gamma}.$$

\end{proof}

\begin{lemma}\label{approximation_factors}
   There exist parameters for \cref{alg3} for which the approximation factor obtained by the algorithm \cref{1} is 188 for XOS value functions with exact demand oracle and 468 for submodular value functions with only value oracle.
\end{lemma}
\begin{proof}
    For the case of $a=1$ demand oracle, i.e., exact demand oracle, when setting the parameters to $m=1.707$, and $M=3.414$, for $\delta = 0.495$, we achieve an approximation factor of $2^{2 \gamma}\cdot 187.82$, which is bounded by $188$ for sufficiently small $\gamma$. 
    This implies a sub-188 ratio between optimal welfare and optimal utility for any $n$-agent setting with XOS value function.

    Recall that for a submodular value function, an $a$-demand oracle with $a=1-\frac{1}{e}$ can be efficiently simulated using only value queries \cite{submodular-demand1, submodular-demand2}.
    For the case of $a=1-\frac{1}{e}$, when setting the parameters to $m=1.713$, and $M=3.399$, for $\delta = 0.4982$, we achieve an approximation factor of $2^{2\gamma}\cdot 467.98$, which is bounded by $468$ for sufficiently small $\gamma$. 

\end{proof}

\subsection{Improving the Benchmark}
\label{sub:XOS-stronger-benchmark}

We have shown a utility approximation algorithm under the strengthened benchmark of optimal welfare. 
We now observe it is possible to provide a constant approximation of utility even relative to the benchmark of optimal $b$-feasible welfare for any $b=2-\Theta(1)$. 
We note that the resulting constant bound on the welfare-utility gap does not extend to arbitrary values of $b$. This is since, as shown in \cref{appendix:subadditive:ub}, for $b=\omega(1)$ the ratio between optimal welfare subject to $b$-feasibility and optimal utility may be unbounded (i.e., may grow with $n$).

\begin{proposition}\label{XOS<2}
    There exists an algorithm that, for any $1\le b \le 2-\Theta(1)$, outputs a set whose utility is an $O(1)$-approximation of the optimal welfare subject to $b$-feasibility, for any $n$-agent setting with an XOS value function. 
      The algorithm runs in polynomial time in $n$ using value and $a$-demand oracles.
\end{proposition}

The idea behind this extension is the observation that the optimal set $S$ needs only to fulfill the conditions of \cref{lemma:alg2} for the proof to go through. We use the fact that its shares are bounded to argue that if there is an agent that is \enquote{expensive}, then any other agent is \enquote{cheap}, and that a \enquote{cheap} agent's share is bounded away from $1$ by a constant to obtain the approximation bounds. This holds for any set whose share is bounded away from $2$ as well, by restricting the range of our analysis parameter $\delta$. 

\begin{proof}[Proof of \cref{XOS<2}]
    Observe that in the proof of \cref{fundamental}, we do not directly use the fact that $S$ is feasible, but use two implications of its feasibility: (1)~$\rho(S)\le \left(2-\frac{1}{8m}\right)^2$; and (2)~since $\rho(S)\le 1$, if any one of the agents $i\in S$ is $\delta$-expensive, then any other agent $j\in S$ is $\delta$-cheap, so the share of $j$ is bounded away from 1.

    Here we will show that this can be extended to any set $S$ whose share $\rho(S)$ is upper bounded by some $2-\eta$ where $\eta$ is a small constant.
    
    The first implication generalizes to any set with $\rho(S)\le \left(2-\frac{1}{8m}\right)^2$, and in particular to any $b$-feasible set where $b\le 2-\eta$.  
    
    As for the second implication, this generalizes to a weaker guarantee: for any constant $\delta$ such that $0\le \delta\le \eta - \Theta(1)$, if agent $i\in S$ is $\delta$-expensive (i.e. $\rho(i)\ge (1-\delta)$), then the share of any other agent $j\in S$ is bounded above by $1$ minus some constant. Indeed, this is the case as
    $$\rho(j)\le \rho(S)-\rho(i)\le (2-\eta)-(1-\delta)=1-(\eta-\delta),$$ 
    which is bounded away from $1$ by a constant depending on our choice of $\delta$.
\end{proof}

We strengthen this result further in \cref{subsection:value:corollaries}, providing an approximation of utility subject to $b$-feasibility for any $0<b\le 2-\Theta(1)$.

%% file: GG_sublinear.tex
\section{Welfare Approximation for sXOS }
\label{section:sublinear}

The welfare approximation algorithm for XOS in Section~\ref{section:XOS} achieves both a constant-factor approximation and an upper bound on the welfare-utility gap, simultaneously. It is also interesting to ask what a direct approximation algorithm can achieve. In this section we demonstrate that a direct algorithm can get improved guarantees, by focusing on symmetric XOS instances and showing a tight $(2+o(1))$-approximation algorithm for welfare.

\begin{theorem}\label{4}
    There exists an algorithm that obtains a $(2+o(1))$-approximation of the optimal welfare, for any $n$-agent setting $(f, c)$ with an sXOS value function and uniform costs. The algorithm uses $O(\log n)$-many value oracle queries.
\end{theorem}

Recall that for symmetric functions, sets are described solely by the \emph{number} of agents in the set. 
We refer to a number as \emph{feasible} if it represents a group size that is feasible in the symmetric instance. 

\begin{proof}[Proof of Theorem~\ref{4}] 
    For any sXOS value function and for any feasible number $k$ we have
    \begin{align*}
        1 &\ge \rho(k) = \frac{c k}{m(k)} \ge ck \cdot \frac{k}{f(k)} = c \cdot \frac{k^2}{f(k)},
    \end{align*}
    where we use the marginals property $m(k)\le \frac{1}{k}f(k)$ (see \cref{sXOS:marginals}). This implies 
    \begin{equation}
        \label{eq:eligible}
        f(k) \ge ck^2.
    \end{equation}
    We say that $k$ is \emph{eligible} if it satisfies Equation~\eqref{eq:eligible}. Note that every feasible number is eligible, but there might be infeasible numbers that are eligible.
    By applying the marginals property, we see that
    \[
        f(k) = f(k-1) + m(k) \le f(k-1) + \frac{f(k)}{k},
    \]
    and hence 
    \begin{equation}
        \label{eq:eligible2}
        f(k-1) \ge f(k) \cdot \frac{k-1}{k} \ge c k(k-1) \ge c(k-1)^2\enspace,
    \end{equation}
    i.e., if $k$ is eligible so is $k-1$.    
    
    Let $n'$ denote the largest eligible number. Due to Equation~\eqref{eq:eligible2}, the set of the eligible numbers is $[n'] = \{1,\ldots,n'\}$. Note that $n'$ can be found using binary search. The set $[n']$ contains all feasible numbers. We can test if any feasible number exists by checking eligibility of $f(1) \ge c$, since $1$ is eligible if and only if it is feasible.  For the rest of the proof, we assume that $n' \ge 2$ and that there exists a feasible number.
    
    The value $f(n') \ge f(S^*)$ is an upper bound on the welfare of the optimal feasible set $S^*$. We now show that we can find a feasible set whose welfare approximates $f(n')$. Let $k^* \in [n']$ be the largest feasible number. Since all numbers $k = n', \ldots,k^*+1$ are infeasible, they satisfy $m(k) < ck$. Hence, for the largest feasible number $k^*$ we see that
    \begin{align*}
    f(k^*) &= f(n') - (f(n') - f(k^*))
    \\&   = f(n') - ( m(n') + m(n'-1) + .. +m(k^*+1) )
    \\&   > f(n') - c(n' + (n'-1) + .. + (k^*+1))
    \\&   = f(n') - \frac{c}{2}\cdot (n'(n'+1)-k^*(k^*+1)) \quad := b(k^*).
    \end{align*} 
    More generally, for any feasible number $k$ with $f(k) \ge b(k)$ we observe that
    \[
        \frac{c}{2}\cdot (n'(n'+1)-k(k+1)) \; \le \; \frac{c}{2}\cdot(n'(n'+1)) \; \le \; \left(\frac 12 + \frac{1}{2n'}\right) f(n'), 
    \]
    since $f(n') \ge c\cdot(n')^2$. For the same reason, 
    \[
        w(k) \; = \; f(k) - ck \; \ge \; b(k) - ck \; \ge \; \left( \frac 12 - \frac{1}{2n'}\right) f(n') - cn' \; \ge \; \left( \frac 12 - \frac{1}{2n'} - \frac{1}{n'} \right) f(n'),
    \]
    i.e., every feasible number $k$ with $f(k) \ge b(k)$ achieves an approximation ratio of $2+o(1)$.

    Finally, let us discuss how to find a feasible number $k \in [n']$ with $f(k) \ge b(k)$ using logarithmic-many value queries. Suppose we have a number $i \in [n']$ with $f(i-1) < b(i-1)$ and $f(i) \ge b(i)$. Then the marginal satisfies
    \[
        m(i) = f(i) - f(i-1) \ge b(i) - b(i-1) = c\cdot i,
    \]
    so $i$ is a feasible number. Moreover, since $f(i) \ge b(i)$, the number $i$ achieves a ratio of $2+o(1)$.

    Initially $f(n') \ge b(n') = f(n')$ and $f(0) = 0 < b(0)$ since $n' \ge 2$. For any $i < j$ such that $f(i) < b(i)$ and $f(j) \ge b(j)$, there must at least one number $k$ in the interval $[i+1,j]$ where $f$ grows above $b$, i.e., $f(k-1) < b(k-1)$ and $f(k) \ge b(k)$. This number $k$ is feasible and achieves a ratio of $2+o(1)$. Thus, for any $i < j$ we can apply binary search to find such a number: For $k = \lfloor (i+j)/2 \rfloor$ we check if $f(k) \ge b(k)$. If so, then we proceed in the interval $[i,k]$, otherwise proceed in the interval $[k,j]$. 
\end{proof}

It remains to show that the approximation factor of \cref{4} is essentially tight.

\begin{proposition}\label{2_is_tight}
    There exists no sublinear-time algorithm (in $n$) with access to only a value oracle that computes an approximation of the optimal welfare with a ratio better than $2$ for all multi-agent settings with an sXOS value function and uniform costs.
\end{proposition}

\begin{proof}
    Consider the function $f(k)=k^2+(1-\varepsilon)k+(n^2+n)$.
    The marginals of $f$ for $k>1$ are $m(k)=2k-\varepsilon$, and the shares are $\rho(k)=\frac{c(k)}{2k-\varepsilon}$.
    We set uniform costs $c=2$. Then for every $k>1$ we observe $\rho(k)=\frac{2k}{2k-\varepsilon}=\left(1+\varepsilon\cdot \frac{1}{2k-\varepsilon}\right)>1$, so $[k]$ is infeasible. 
    The only feasible set is $[1]$, and $w(1)=1-\varepsilon+n^2+n$.

    Note that for any constant $\delta > 0$, no sublinear-time algorithm can query all the values in the interval $[(1-\delta)n, n]$.
    Thus, for any sufficiently large $n$, there must be a number $k^*$ in this segment for which the algorithm did not apply a value query. We modify $f(k^*)$ to be $f(k^*)+\varepsilon$, which makes the set $[k^*]$ feasible. Then this set will be a set with optimal welfare.

    Then the ratio between $w(k^*)$ and $w(1)$ becomes 
    \[
    \frac{(k^*)^2+(1-\varepsilon)k^*+(n^2+1)+n}{1-\varepsilon+n^2+n}\ge \frac{((1-\delta)^2+1)n^2+(1-\delta)n}{n^2+n}=\frac{((1-\delta)^2+1)+\frac{1}{n}}{1+\frac{1}{n}}
    \]
    which approaches 2 when (the constant) $\delta\to 0$ and $n\to \infty$.

    Finally, let us argue that $f$ is sXOS, which will conclude the proof.
    Recall that $f$ is sXOS if it satisfies the three sXOS properties (see \cref{sXOS:main:prop}). The only non-trivial one is $\frac{1}{j}f(j)\le \frac{1}{j-1}f(j-1)$ for every $j>1$, which is equivalent to $j\cdot m(j+1)\le f(j)$. This will hold throughout if it holds for $j+1=k^*$, since we increased $f(k^*)$. Note that $m(j+1)=m(k^*)=2k^*$, and thus $j\cdot m(j+1)=2j(j+1)=2j^2+j$, while $f(j)=j^2+(1-\varepsilon) j+(n^2+1) \ge j^2 + j + n^2$. Thus the property holds if $n^2\ge j^2$, which is true for every $j$.

    While the argument above works directly for deterministic algorithms, it can be extended to randomized algorithms by choosing $k^*$ uniformly at random in the range $[(1-\delta)n, n]$. When querying a value $k \neq k^*$, an algorithm does not learn anything about the location of $k^*$ (except, of course, that $k \neq k^*$). By Yao's principle, we can assume $k^*$ is chosen uniformly at random, and then any deterministic algorithm can find $k^*$ only with a probability of $o(1)$. Thus, the overall ratio achieved is at most $2-o(1)$.   
\end{proof}

%% file: AZ_value.tex
\section{Beyond Welfare: Value Approximation}
\label{section:value}

In this section we extend our welfare results to the context of \emph{value} maximization subject to feasibility.
Unlike the welfare-utility gap, the value-welfare gap can be infinite even in symmetric instances:

\begin{observation}
\label{Value-Walfare-Gap}
    The value-welfare gap can be unbounded even for a single agent.
\end{observation}

\begin{proof}
    Consider a single agent whose value equals its cost, and both equal $1$. In this case, $n=1$ and $\rho([n])=1$, i.e., $[n]$ is feasible. Thus, the optimal value is $f([n])=1$, while the optimal welfare (and utility) is $0$. 
\end{proof}

In light of this observation, we show that \cref{alg3} finds a set whose \emph{value} approximates the optimal value. The algorithm works for any $b$-feasibility constraint. 

\begin{proposition}
\label{thm:XOS_value}
There exists an algorithm that, for any finite $b>0$, outputs a $b$-feasible set whose value is an $O(1)$-approximation of the optimal value subject to $b$-feasibility, for any $n$-agent setting with an XOS value function. The algorithm runs in time polynomial in $n$ using value and $a$-demand oracles, and guarantees a $(\frac{3}{A}+o(1))$-approximation where $A=(a/4)^2$.
\end{proposition}
\cref{thm:XOS_value} yields a $(48+o(1))$-approximation given an exact demand oracle---a factor $4$ improvement over our utility and welfare approximation results. 
For submodular value functions, $a=(1-\frac{1}{e})$ so we obtain an approximation of $121+o(1)$ using only value oracle access. 
A similar result extends to any $B$-transfer constraint (see Proposition~\ref{pro:XOS_transfer}).

In \cref{subsection:value:corollaries} we show 
corollaries of \cref{thm:XOS_value}
for approximate welfare and utility under $b$-feasibility for any $0<b\le 2-\Theta(1)$.
In \cref{value:symmatric} we discuss value approximation in the symmetric setting and demonstrate how our welfare results translate to value.

\subsection{Proof of Proposition~\ref{thm:XOS_value}} 
\label{sub:value-approx}

In this section we show that \cref{alg3} satisfies the requirements of \cref{thm:XOS_value}, by finding a feasible set whose value is an $O(1)$-approximation of the optimal value subject to feasibility.
Our main tool is \cref{lem:alg2:value}, which is a special case of \cref{lemma:alg2}.
In order to prove \cref{thm:XOS_value}, we first prove it for the case of $b=1$, then extend the result to any $b>0$ using \cref{obs:value_invariant}.

\paragraph{The case of $b=1$.} 

In this case we show that \cref{alg3} satisfies the requirements of \cref{thm:XOS_value}, by finding a feasible set whose value is an $O(1)$-approximation of the optimal value subject to feasibility.

\begin{lemma}
\label{lem:XOS_value_helper}
\cref{thm:XOS_value} holds for $b=1$.
\end{lemma}

To prove \cref{lem:XOS_value_helper} we use \cref{alg2}. The following corollary states useful guarantees of \cref{alg2};
it is a special case of \cref{lemma:alg2} in which the domain of $f$ is assumed to be restricted to $D_y$, and the set $S$ is assumed to be feasible. 

\begin{corollary} 
\label{lem:alg2:value}
    For any $n$-agent setting $(f,c)$ with an XOS value function, a target value $y\in[0, f([n])]$, and a domain 
    $D_y=\{i\in [n]: f(i)\le \frac{A}{m}\cdot \frac{y}{M}\}$ where $M>1$ and $m\ge 1$. 
    There exists a polynomial-time algorithm with value and $a$-demand oracle access to $f$ that guarantees the following. Given an input $(f_{|D_y}, c_{|D_y}, y, m)$, if there exists a feasible set $S\subseteq D_y$ such that $y\le f(S)$, then the algorithm returns a set $U \subseteq D_y$ that satisfies:
    \begin{enumerate}[nosep,noitemsep]
        \item $\frac{A}{m}\cdot y\ge f(U)\ge \left(1-\frac{1}{M}
        \right)\cdot \frac{A}{2m}\cdot y$;
        \item $\rho(U)\le \frac{1}{m}$.
    \end{enumerate}
\end{corollary}

\cref{lem:alg2:value} essentially states that if there exists a feasible set $S$ with value at least $y$ containing only agents with ``small'' individual values (relative to $y$), then \cref{alg2} will output a feasible set $U$ whose value is a constant approximation of $y$. 

Recall that \cref{alg3} performs \emph{doubling} over possible values of $y$ with a multiplicative jump of $2^\gamma$ (where $\gamma$ is some small constant). 
For each value $y$, it runs \cref{alg2} and adds the resulting set to a collection that also includes all the singletons of individual agents. 
We claim that the collection generated by this process contains a feasible set whose value is a $48+o(1)$ approximation of optimal value: 

\begin{proof}[Proof of \cref{lem:XOS_value_helper}]
    We set $m=1$ and $M=3$. Let $S$ be the optimal value set, and consider $y$ to be the estimator of $f(S)$; that is, $f(S)\cdot 2^{-\gamma}<y\le f(S)$.
    If $S$ contains only agents with individual values smaller than $\frac{A}{3}y$, then by \cref{lem:alg2:value}, the collection generated by \cref{alg3} will include a set which is an $\frac{3}{A}$-approximation 
    of $y$, implying a $\frac{3}{A}\cdot 2^{\gamma}$ approximation of $f(S)$.
    On the other hand, if $S$ contains an agent $i$ for which $f(i)\ge \frac{A}{3}y$, then it alone provides the desired approximation.
\end{proof}

\paragraph{Remark.} 
In the proof of \cref{lem:XOS_value_helper} we either have a set $T$ whose value is a constant approximation of optimal value such that its share is upper bounded by $\frac{1}{m}$, or we must have a single agent whose value is a constant approximation of optimal value (and its share is not necessarily bounded). By plugging in any $m\ge 1+\Theta(1)$, we see that in any $(f, c)$ instance, either there exists a set whose utility is a constant approximation of the optimal value, or there exists a single agent whose value is very large while its share is very small. This shows that the instance shown in \cref{Value-Walfare-Gap} of a single agent whose value equals its cost is essentially the only instance with an unbounded value-utility ratio.

\paragraph{General $b$.} 

We observe that the problem of maximizing value is invariant to the value of $b$ in the $b$-feasible constraint. This along with \cref{lem:XOS_value_helper} immediately proves \cref{thm:XOS_value}.
\begin{observation}
    \label{obs:value_invariant}
    Let $(f, c)$ be an $n$-agent setting and let $b, b' > 0$ be two positive feasibility constraints.
    A set $S$ is $b$-feasible in $(f, c)$ if and only if $S$ is $b'$-feasible in $(f, \frac{b'}{b} c_i)$.
\end{observation}

\begin{proof}
\begin{equation*}
    \sum_{i\in S} \frac{c_i}{f(i:S)}\le b \iff \sum_{i\in S} \frac{\frac{b'}{b} c_i}{f(i:S)}\le b'.
\end{equation*}

\end{proof}

\begin{proof}[Proof of \cref{thm:XOS_value}]
    Combining \cref{lem:XOS_value_helper}, which proves the proposition for $b=1$, and \cref{obs:value_invariant}, which shows that the problem of maximizing value is invariant to the value of $b$, the proposition follows.
\end{proof}

\paragraph{$B$-transfer constraint.}

We establish a similar result for the $B$-transfer constraint. We use the fact that any non-empty set $S$ satisfies the $B$-transfer constraint if and only if it is $b$-feasible for $b=\frac{B}{f(S)}$, and the fact that any set $T$ whose value is a constant approximation of $f(S)$ such that $f(T)\le f(S)$ also satisfies the $B$-transfer constraint. 

\begin{proposition}\label{pro:XOS_transfer}
    There exists an algorithm that, for any $b>0$ and $B > 0$, outputs a set satisfying the $b$-feasibility and $B$-transfer constraint, whose value is an $O(1)$-approximation of the optimal value subject to the $b$-feasibility and $B$-transfer constraint, for any $n$-agent setting with an XOS value function.
    The algorithm runs in polynomial time in $n$ using value and $a$-demand oracles.
\end{proposition}

In order to prove \cref{pro:XOS_transfer}, we will use a slightly stronger version of \cref{lem:alg2:value}, which weakens the condition that the set $S$ must be feasible ($\pay(S)\le 1$) to the condition $\pay(S)\le 2$. Observe that the strengthened version is still a special case of \cref{lemma:alg2}, which covers any set $S$ with $\pay(S)\le (2-\frac{1}{8m})^2$.

\begin{proof}[Proof of \cref{pro:XOS_transfer}]
    Let $S$ be the optimal value set. Suppose $S \neq \emptyset$ (otherwise, an approximation is trivial).
    Let $b(S)=\min\{b, \frac{B}{f(S)}\}$.
    First, note that if we had known the value of $S$, i.e., had known $f(S)$,
    we would have known $b(S)$. Then, with a similar procedure as in the one from the proof as of \cref{thm:XOS_value} (i.e., cost scaling to $\tilde c = c/b(S)$, and then using \cref{lem:XOS_value_helper}), we would have been able to obtain a $b(S)$-feasible set $T$ whose value is a constant approximation of $f(S)$ and is upper bounded by $f(S)$, making it $b$-feasbile and $B$-transfer. 
    Since we do not know $b$ a priory, we will approximate it with doubling. 
    
    Let $y$ be an estimator of $f(S)$, i.e., $f(S)\cdot 2^{-\gamma}<y\le f(S)$, where $\gamma>0$ is some small constant.
    Let $\hat b$ be an estimator of $b(S)$, i.e., 
    $b(S)\cdot 2^{-\mu} < \hat b \le b(S)$, where $1\ge \mu > 0$. 
    Both estimations can be obtained by repeated doubling.
    
    By cost scaling to $\tilde c = c/\hat b$, we see that the scaled share is
    $\tilde \rho(S) = \rho(S)/\hat b \le b(S)/\hat b< 2^{\mu}\le 2$.
    Then, similarly to \cref{lem:XOS_value_helper}, and by the strengthened \cref{lem:alg2:value}, in the scaled instance we obtain either an agent $i\in S$ such that $f(S)\ge f(i)> \frac{A}{3}\cdot y$ or a feasible set $U$ such that $A\cdot y\ge f(U)\ge \frac{A}{3}\cdot y$.
    In both cases, the set is a $\frac{3}{A}\cdot 2^\gamma$ approximation of $f(S)$ and its value is upper bounded by $f(S)$, meaning that it also satisfies both $b$-feasibility and $B$-transfer constraints.
\end{proof}


\subsection{Approximations of Welfare and Utility Under Generalized Constraints}
\label{subsection:value:corollaries}

In this subsection, we show approximation algorithms for welfare and utility beyond standard feasibility. 
Our starting point is the following observation:

\begin{observation}
\label{obs:all-objs}
    Under a $b$-feasibility constraint with $b\le 1-\Theta(1)$, it holds that value, utility, and welfare of any $b$-feasible set differ by at most a constant factor.
\end{observation}

\begin{proof}
For any $S$ such that $b$ upper-bounds $\rho(S)$,
    $$
    (1-b)f(S) \le (1-\rho(S))f(S) = g(S) \le w(S) \le f(S).
    $$
\end{proof}

We obtain the following corollary of \cref{thm:XOS_value}:

\begin{corollary}\label{b-welfare}
  There exists an algorithm that, for any $0<b \le 2-\Theta(1)$, outputs a $b$-feasible set, whose utility is an $O(1)$-approximation of the optimal welfare subject to $b$-feasibility, for any $n$-agent setting with an XOS value function. 
  The algorithm runs in polynomial time in $n$ using value and $a$-demand oracles.
\end{corollary}

\begin{proof}
    We prove this result with two separate cases. 
    
    (1) For the case of $0<b\le 1-\Theta(1)$, by \cref{thm:XOS_value} we can obtain a $b$-feasible set whose value is a $\frac{3}{A}\cdot 2^\gamma$ approximation of the optimal value subject to $b$-feasibility. 
    Then using~\cref{obs:all-objs}, this set's utility is a $(\frac{1}{(1-b)}\cdot \frac{3}{A}\cdot 2^\gamma)$-approximation of the optimal value. 

    (2) Let $S$ be the optimal welfare set subject to $b$-feasibility. For the case of $1\le b<2-\Theta(1)$, we already provided a constant-factor approximation of $w(S)$ in \cref{XOS<2}. This result can be extended to the range of $\Theta(1)\le b\le 1$ as well by restricting the parameter $m$ to be such that $\frac{1}{m}\le b$. 
    With this restriction on $m$, we ensure that \cref{alg3} provides either an agent $i\in S$ whose utility is a constant approximation of $w(S)$ (and $\rho(i)\le \rho(S)\le b$), or a set $T$ whose value is a constant-factor approximation of $f(S)$, and whose share is a constant that is upper-bounded by $\frac{1}{m}\le b$. 
\end{proof}

Notice that for $b\le 1-\Theta(1)$, the approximation in our proof establishes a guarantee directly against
the stronger benchmark of optimal value. Conversely, when $b$ approaches 1, we leverage the approximations from \cref{1} and \cref{XOS<2}. Specifically, for exact demand oracles (i.e. $a=1$), we have 
an approximation of $\frac{48+o(1)}{1-b}$ for optimal value. Additionally, we provide a $188$ approximation of optimal welfare with the parameter $m=1.713$, i.e., for $b\ge \frac{1}{m}\approx 0.58$ (see \cref{fundamental}). This implies a factor $4$ 
increase for small values of $b$ compared to $b=1$.

\begin{corollary}\label{B-welfare}
    There exists an algorithm that, for any $0<b\le 2-\Theta(1)$ and $B > 0$, outputs a $b$-feasible and $B$-transfer set whose utility represents an $O(1)$-approximation of optimal welfare subject to $B$-transfer and $b$-feasibility, for any $n$-agent setting with XOS value function.
    The algorithm runs in polynomial time using value and $a$-demand oracles.
\end{corollary}

\begin{proof}
    For convenience, for any $S\neq \emptyset$ let $b(S)=\min\{b, \frac{B}{f(S)}\}$. Thus, finding a set that is both $ B$-transfer and $b$-feasible amounts to finding a set $S$ which is $b(S)$-feasible.

    Consider the optimal welfare set $S_w$ subject to $b$-feasibility and $B$-transfer.
    Henceforth, we assume that $S_w$ is not the empty set, since otherwise an approximation is trivial.
    
    An important observation is that for any set $S$ that is $b(S)$-feasible, if a set $T$ satisfies both $f(T)\le f(S)$ and $\rho(T)\le b(S)$, then $T$ is $b(T)$-feasible.  
    This implies that in order to find a set $T$ which is $b(T)$-feasible, it suffices to find a set $T$ that satisfies both $f(T)\le f(S_w)$ and $\rho(T)\le b(S_w)$.
    In particular, any agent $i\in S_w$ is $b(i)$-feasible.

    Note that if we had known the value of $b(S_w)$, then we could have used \cref{b-welfare} directly to obtain our desired approximations. 
    Since we do not know $b(S_w)$ a-priory, similar to \cref{pro:XOS_transfer}, we will estimate $b(S_w)$ by performing a doubling procedure. 
    Let $\hat b$ be the estimator of $b(S_w)$ obtained by this procedure; that is, $b(S_w)\le \hat b < b(S_w)\cdot 2^\mu$ for some small constant $\mu>0$.

    \paragraph{The case where $\hat b\le 1 - \Theta(1)$.}
    In this case we are guaranteed that $b(S_w)\le 1-\Theta(1)$ as well.
    We set $m$ to be such that $\frac{1}{m}= 2^{-\mu}$ and scale the costs by $\frac{1}{\hat b}$.
    Proceeding exactly as in \cref{b-welfare}, we are able to find one of the following:
    \begin{enumerate}
        \item An agent $i\in S_w$ such that $f(i)\ge \frac{3}{A}\cdot 2^{-(\delta+\mu)}\cdot f(S_w)$ and $\rho(i)\le \rho(S_w)\le b(S_w)$.  Thus $g(i)=(1-\rho(i))f(i)=\Theta(1) f(S_w)$, i.e., $g(i)$ is a constant approximation of $f(S_w)$.
        In this case, $i$ alone satisfies all our desired conditions (i.e., its utility is a constant approximation of $f(S_w)$ and it is $b(i)$-feasible, hence $b$-feasible and $B$-transfer).

        \item A set $T$ such that $\frac{A}{3}\cdot 2^{-(\delta+\mu)}\cdot f(S_w)\le f(T)\le f(S_w)$, and whose (scaled) share is $\tilde \rho(T)\le 2^{-\mu}$, meaning that $\rho(T)\le \hat b\cdot 2^{-\mu}<b(S_w)$. Since $T$ satisfies both $f(T)\le f(S_w)$ and $\rho(T)<b(S_w)$, it is $b(T)$-feasible. 
        Thus, in this case, $T$ satisfies our requirements.
    \end{enumerate}

    \paragraph{The case of $\Theta(1)\le \hat b \le 2-\Theta(1)$.}
    In this case we are guaranteed that $b(S_w)\ge \Theta(1)$.
    Similarly to \cref{b-welfare}, we constraining $m>1$ to be such that $\frac{1}{m}\le \hat b\cdot 2^{-\mu}\le b(S_w)$.
    Note that $\frac{1}{m}$ is not constrained to be arbitrarily small since $\hat b\ge \Theta(1)$.    
    With this constraint on $m$ and since $\hat b \le 2-\Theta(1)$ we are guaranteed by \cref{XOS<2} that \cref{alg3} will either: 
    (1) provide an agent $i\in S_w$ whose utility is a constant-factor approximations of $w(S_w)$; or
    (2) provide a set $T$ whose value is a constant approximation of $w(S_w)$, such that $f(T)\le f(S_w)$ and whose share  $\rho(T)\le \frac{1}{m}$ (which is smaller then $b(S_w)$ and also is bounded away from $1$ by a constant).
    In both cases the set satisfies the $b$-feasibility and $B$-transfer constraint as desired.
\end{proof}

The analysis of \cref{B-welfare} can be strengthened by additionally considering estimators of $\rho(S_w)$ and $\rho(S_v)$, where $S_v$ denotes the optimal value set subject to the $b$-feasibility and $B$-transfer constraints.
Similarly to \cref{b-welfare}, for $\rho(S_v) \le 1-\Theta(1)$, the above proof guarantees a direct approximation relative to the stronger benchmark of optimal value.
Furthermore, for $a=1$ (as well as $a=(1-1/e)$), we see a factor $4$ improvement between small values of $b$ and $b=1$.

\subsection{The Symmetric Case}
\label{value:symmatric}

We show that for sXOS value functions the welfare-value gap diminishes with set size for any $b$-feasibility constraint.

\begin{proposition}
    \label{w=f}
   For any $n$-agent setting $(f, c)$ with an sXOS value function, it holds that
   $$f(k)\ge w(k)\ge \left(1-\frac{b}{k}\right)f(k)$$
   for any non-empty, $b$-feasible set $[k]$. 
\end{proposition}

To see this, we use the following observation, by which for sXOS and any $b$, the cost of a set diminishes with size:

\begin{observation}
\label{sXOS:prop:welfare:LB}
For any $n$-agent setting $(f, c)$ with an sXOS value function, for any $k\in [n]$, if $[k]$ is $b$-feasible, then 
$$c(k)\le \frac{b}{k}\cdot f(k).$$
\end{observation}
\begin{proof}
    Since $[k]$ is $b$-feasible, and from the marginals property (see \cref{sXOS:marginals}), it holds that
    $$
        b\ge \rho(k)=\frac{c(k)}{m(k)}\ge c(k)\cdot \frac{k}{f(k)}.
    $$
    By rearranging, we obtain $c(k)\le \frac{b}{k}\cdot f(k)$.
\end{proof}

\begin{proof}[Proof of \cref{w=f}]
    The first inequality is immediate from the definition of welfare. Consider the second inequality. For any feasible set $[k]$, by \cref{sXOS:prop:welfare:LB}, we have that $c(k)\le \frac{b}{k}f(k)$. It follows that 
    $$w(k)=f(k)-c(k)\ge \left(1-\frac{b}{k}\right)f(k).$$
\end{proof}

As a result, it is straightforward to translate our optimal welfare approximation results for the symmetric case to optimal value. In turn, since value results are invariant to the $b$-feasibility constraint, and since for any $b\le 1-\Theta(1)$, value, welfare, and utility are a constant factor apart, these results directly imply constant-factor approximation results for welfare and utility subject to $b$-feasibility.

%% file: AZ_acknowledgement.tex
\section*{Acknowledgments}
We are deeply grateful to Hadar Strauss, whose invaluable contributions were essential in shaping this work.

Gil Aharoni and Inbal Talgam-Cohen were supported by the European Research Council (ERC) under the European Union's Horizon 2020 research and innovation program (grant agreement No.~101077862), by the Israel Science Foundation (grant No.~3331/24), by the NSF-BSF (grant No.~2021680), and by a Google Research Scholar Award.

Martin Hoefer was supported by DFG Research Unit ADYN (project number 411362735) and grant Ho 3831/9-1 (project number 514505843).

We thank Yossi Azar, Shiri Ron and Maya Schlesinger for helpful discussions and feedback. We further thank Shimi Cohen, Ido Mor, and Shaul Rosner for general feedback and advice.

%% file: KK_appendix.tex
\newpage

\section{Standard Definitions} 
\label{Appendix: perliminaries}

In this section, we provide standard definitions for set function classes.
Let $f:2^{[n]}\to \mathbb R_+$ be a value function, where $[n]$ represents the universe of agents, and for every $S\subseteq [n]$ the value $f(S)$ represents the value generated by the set of agents $S$. The following classes of value functions are well-studied:

\begin{enumerate}
    \item \emph{Additive}:
    $f$ is additive if $f(S \cup \{i\}) = f(S) + f(i)$ for every $S\subseteq [n], i \in [n] \setminus S$.
    \item \emph{Submodular}: $f$ is submodular if $f(i : S) \ge f(i : S \cup \{j\})$ for every $S$ and $i, j \in [n]\setminus S$.
    \item \emph{XOS} (or fractionally subadditive): $f$ is XOS if there exists a finite family $I$ of additive functions $\{\al_i\}_{i\in I}$ such that $f(S) = \max_{i\in I}\, \al_i(S)$ for every $S\subseteq [n]$.
    \item \emph{Subadditive}: $f$ is subadditive if $f(S \cup T) \le f(S) + f(T)$ for every $S, T\subseteq[n]$.
    \item \emph{Supermodular}: $f$ is supermodular if $f(i : S) \le f(i : S \cup \{j\})$ for every $S$ and $i, j \in [n]\setminus S$.
\end{enumerate}

It is well-known that these classes of set functions form a hierarchy~\cite{Hierarchy1}: 
\[\text{additive}\subsetneq\text{submodular}\subsetneq{XOS}\subsetneq\text{subadditive}.\]
We note that all strict inclusions apply even when restricting to symmetric set functions.

\input{KKA_XOS_properties}

\input{KKB_algorithm2}

\section{Unconstrained Welfare}
\label{appendix:subadditive:ub}

In this section, we analyze unconstrained welfare, i.e., $\max_{S\subseteq [n]} w(S)$ and site relation to utility. We show a simple $n$-approximation algorithm for utility with respect to the benchmark of unconstrained welfare for any subadditive function. We show that the approximation factor matches the gap between unconstrained welfare and utility even for additive functions, by establishing a tight lower bound of $n$ on this gap.

We first present a simple algorithm for subadditive functions $f$ that finds a set whose utility is an $n$-approximation of unconstrained welfare. Consider the agent $i$ which has the largest individual utility value $g(i)$ 
(note that $g(i) = w(i)$ for every single agent $i$). This single agent represents an $n$-approximation of $\max_{S\subseteq [n]} w(S)$, since by subadditivity, for the set $T$ maximizing $\max_{S\subseteq [n]} w(S)$ it holds that
    \begin{align*}
        w(T)\le \sum_{j\in T}w(j)\le |S|\cdot g(i)\le n \cdot g(i).
    \end{align*}
As a result, the ratio between optimal welfare and optimal utility is upper bounded by $n$. For symmetric subadditive instances with uniform costs, it even suffices to query a single value of $f(1)$ (to evaluate whether a single agent is feasible) to obtain an $n$-approximation. 

\begin{observation} \label{subadditive n-UB}
    For any subadditive value function, the ratio between unconstrained welfare and optimal utility is upper bounded by $n$.
    Moreover, for any instance such that the unconstrained welfare is upper bounded by $M\cdot OPT(w)$ for some $M\ge 1$, the ratio between the optimal welfare and optimal utility is upper bounded by $\frac{n}{M}$.
\end{observation}

\begin{observation}\label{subadd:approx:logUB}
    For any $n$-agent setting with a symmetric subadditive value function and uniform costs, there exists a constant-time algorithm with access to a value oracle that computes a set whose utility is an $n$-approximation of unconstrained welfare.
\end{observation}

We will now show an example with $n$ agents and a symmetric additive value function $f$. In this example, only singleton teams are feasible, meaning that the optimal utility and optimal (feasible) welfare coincide but are bounded away from the unconstrained welfare by a factor of $n$.

\begin{observation}\label{demand_set:tight}
    Even for symmetric additive value functions, the gap between unconstrained welfare and utility is lower bounded by $n$.
\end{observation}

\begin{proof}
    Consider the following symmetric additive function such that $f(k)=kn$ for every $i$ and $c(k)=k(n-1)$, and let $n\ge 3$.
    Notice that $\rho(k)=k\frac{n-1}{n}$ which is feasible only for $k=1$.
    Thus $OPT(w)=OPT(g)=w(1)=g(1)=1$.
    On the other hand, $w(n)=f(n)-c(n)=n^2-n(n-1)=n$. As a result, we obtain a ratio of $n$. 
\end{proof}

\section{Further Results for Submodular}

\subsection{Welfare Maximization for Symmetric Submodular}
We show that for symmetric submodular a form of binary search is sufficient to find the set maximizing.

\begin{proposition}
\label{s-submod:binary}
    There exists an algorithm that obtains a set with optimal welfare for any $n$-agent setting with submodular value function. The algorithm runs in logarithmic time using value oracles.
\end{proposition}

    This follows from the fact that there is a clear order for both the welfare $w(k)$ and the shares allocation function $\rho$ as shown in the following two lemmas.

    \begin{lemma}
        For every symmetric submodular function $f$ and $c_1\le ...\le c_n$, there exists some $k$ for which every smaller size is feasible, and every larger size is infeasible.
    \end{lemma}
    \begin{proof}
        This follows directly from the fact that for submodular functions for every two sets $S\subseteq T$, we have $\rho_S(S)\le \rho_T(T)$.
    \end{proof}
    \begin{lemma}
        For every symmetric submodular function $f$ and $c_1\le ...\le c_n$ there exists some $k$ for which $w(1)\le ...\le w(k)\ge ...\ge w(n)$
    \end{lemma}
    \begin{proof}
        We show two claims which together imply the statement of the lemma (by induction) 
    \begin{enumerate}
        \item $w(k)\ge w(k+1)\Longrightarrow w(k+1)\ge w(k+2)$
        \item $w(k)\ge w(k-1)\Longrightarrow w(k-1)\ge w(k-2)$
    \end{enumerate}
    
    First, notice that $w(k)-w(k-1)=m(k)-c_k$ since
    $$w(k)-w(k-1)=(f(k)-c(k))-(f(k-1)-c(k-1))=m(k)-c_k.$$
    
    The first claim holds since
    \begin{align*}
        &w(k+2)-w(k+1)=m(k+1)-c_{k+1}
        \le m(k)-c_k=w(w+1)-w(k)\le 0,
    \end{align*}
    where the first inequality is due to $m(k+1)\le m(k)$ and $c_{k+1}\ge c_k$.

    The second claim holds since
    \begin{align*}
        &w(k-1)-w(k-2)=m(k-1)-c_{k-1}
        \ge 
        m(k)-c_k=w(k)-w(k-1)\ge 0,
    \end{align*}
    where the first inequality is due to $m(k-1)\ge m(k)$ and $c_{k-1}\le c_k$.
    \end{proof}

\subsection{Lower Bound for Welfare Approximation}
\label{gen:submod:LB}

\begin{proposition}\label{e/e-1}
    Any algorithm that uses only a polynomial number of value queries for $n$-agent settings with submodular value function cannot guarantee an approximation ratio for welfare better than $\frac{e}{e-1}$ unless $\classP = \classNP$.
\end{proposition}

The proof relies on the following statement about a promise problem.
\begin{fact}[see Proposition 3.2 in \cite{inapproximability} following \cite{inapproximability:prop}]
For any $M>1$, for $\varepsilon<e^{-1}$, and $k\in \mathbb{N}$ let $f$ normalized unweighted coverage function such that $f(i)=\frac{1}{k}$ it is \classNP-hard to distinguish between the two scenarios:
\begin{enumerate}
    \item There exists a set $S'$ of size $k$ for which $f(S')=1$.
    \item For every $\aL < M$ and every set $S$ of size $\aL k$, we have $f(S)\le 1 +\varepsilon - e^{-\aL}$.
\end{enumerate}
\end{fact}
\begin{proof}[Proof of \cref{e/e-1}]
    We want to show that for every feasible set $S$ in case 2 of the promise problem, the ratio $w(S')/w(S)$ is lower bounded by a constant.

    Notice that for every set $S$ we have $f(i:S)\le f(i)=\frac{1}{k}$, thus for the set $S'$ in case (1) we have $f(i:S')=f(S')-f(S'\setminus i)\ge f(S')-\frac{k-1}{k}\ge \frac{1}{k}$ hence $f(i:S')=\frac{1}{k}$.
    
    We choose a uniform cost function $c$ for which in case 1, $\rho(S')=1$. Recall that $|S'|=k$, thus $\rho(S')=\sum_{i\in S'}\frac{c}{f(i:S')}=k^2c=1$ iff $c=\frac{1}{k^2}$. So $w(S')=f(S')-kc=1-\frac{1}{k}$.
    
    Now since for every set $S$ we have $f(i:S)\le \frac{1}{k}$, in case 2, every $S$ of size larger than $k$ will not be feasible. We are left to analyze the case where $|S|\le k$. In this case $|S|=\aL k$ for some $\aL \le 1$. thus $f(S)\le (1+\varepsilon) -e^{-\aL}\le (1+\varepsilon)-e^{-1}$ which serves as an upper bound for $w(S)$ as well.
    
    We get a ratio $w(S')/w(S)\ge \frac{1-\frac{1}{k}}{(1+\varepsilon)-e^{-1}}$. Taking $k$ to infinity and $\varepsilon$ to zero we approach a lower bound of $\frac{1}{1-e^{-1}}=\frac{e}{e-1}$. 
\end{proof}

%% file: KKA_XOS_properties.tex
\section{Properties of XOS and sXOS}
\label{section:sXOS}

In this section, we discuss representations of XOS and sXOS functions and prove useful properties. 
In \cref{sub:sXOS}, we focus on sXOS, respectively, and in \cref{sub:XOS-sXOS}, we show additional properties of both classes.

\subsection{XOS Canonical Representation}
\label{sub:XOS}

\begin{definition}[XOS canonical representation] 
    \label{frac:canon}
    A family of non-negative additive set functions $\{\al_S\}_{S\subseteq [n]}$ is a \emph{canonical representation} 
    if it satisfies the following to properties
    \begin{enumerate}
        \item For every set $S\subseteq [n]$, the value $\al_S(S)=\max_T \al_T(S)$.
        \item For every $S,T\subseteq[n]$, it holds that $\al_S(T)=\al_S(T\cap S)$.
    \end{enumerate}
\end{definition}

\begin{lemma} \label{lem-canon-equiv}
    A function $f$ is an XOS function if and only if it has a canonical representation. 
\end{lemma}

\begin{proof}
The first direction is follows by definition. For the second direction,
recall that every XOS function $f$ is a maximum over a family of additive functions $\{\be_i\}_{i\in I}$ such that $f(S)=\max_{i\in I}\be_i(S)$. 
Let $\tau:2^{[n]}\to I$ be a mapping from a set to its pointwise maximum; that is, $\be_{\tau(S)}(S)\ge \be_i(S)$ for every $i\in I$.  
Consider the set $\{\al_S\}_{S\subseteq[n]}$ such that $\al_S(T)=\be_{\tau(S)}(T\cap S)$.
We claim that $\{\al_S\}_{S\subseteq[n]}$ is a canonical representation of $f$. 
We need to show that $\{\al_S\}_{S\subseteq[n]}$ is a canonical representation and that $\al_S(S)=f(S)$. 
Observe that $\al_S(S)=\be_{\tau(S)}(S)=\max_{i\in I}\be_i(S)=f(S)$ 
$$
\al_S(T\cap S)=\be_{\tau(S)}((T\cap S)\cap S)=\be_{\tau(S)}(T\cap S)=\al_S(T).
$$
This establishes Property 2 of \cref{frac:canon}. It is left to show Property 1, i.e.~that $\al_S(S)=\max_{T\subseteq[n]}\al_T(S)$. This holds since for every pair of sets $S$ and $T$ it holds that
    \begin{align*}
        \al_T(S)=\be_{\tau(T)}(S\cap T)\le \be_{\tau(T)}(S)\le \be_{\tau(S)}(S)=\al_S(S).
    \end{align*}
\end{proof}

\begin{lemma}[Properties of the XOS canonical representation]
    \label{lem:canon-rep-properties}
    Let $\{\al_S\}_{S\subseteq[n]}$ be a canonical representation of an XOS function. Then,
    \begin{enumerate}
        \item $\al_S(S)\ge \al_S(T)$ for every two sets $S, T\subseteq [n]$.
        \item $\al_S(S)\le \al_{S\cup i}(S\cup i)$ for every $S\subseteq [n]$, $i\in [n]$.
        \item For any $S\subseteq [n]$ there exists $i\in S$ for which $\frac{1}{|S|}\cdot\al_S(S)\le \frac{1}{|S\setminus i|}\cdot\al_{S\setminus i}(S\setminus i)$.
    \end{enumerate}
\end{lemma}

\begin{proof}
    The first two properties follow from the definition.
    We turn to the third property.
    Let $i=\min_{i\in S}\{\al_S(i)\}$. From minimality, $\al_S(i)\le \frac{1}{|S|}\cdot\al_S(S)$.
    Therefore, it follows that
    \begin{align*}
        \al_S(S)
    = \al_S(S\setminus i) + \al_S(i)
    \le \al_S(S\setminus i)+\frac{1}{|S|}\cdot\al_S(S),
    \end{align*}
    where the equality is by additivity of $\al_S$.
    By reordering and from the first XOS canonical representation property, it holds that $(1-\frac{1}{|S|})\cdot\al_S(S)\le \al_S(S\setminus i)\le \al_{S\setminus i}(S\setminus i)$. Since $(1-\frac{1}{|S|})=\frac{|S|-1}{|S|}=\frac{|S\setminus i|}{|S|}$,
    we obtain that $\frac{1}{|S|}\cdot\al_S(S)\le \frac{1}{|S\setminus i|}\cdot\al_{S\setminus i}(S\setminus i)$, as claimed. 

\end{proof}

\subsection{sXOS: Equivalent Representations}
\label{sub:sXOS}

For the specific case of symmetric XOS (sXOS) functions, we will further define specific, stronger representations, which we term the canonical \emph{set/size/function} representations. We show that every sXOS function has a unique representation of each kind and that every valid representation corresponds to a specific sXOS function.

\begin{definition}[sXOS canonical set representation]
    \label{def:canon-set}
    A family of set functions $\{\al_S\}_{S\subseteq [n]}$ is a \emph{canonical set representation} if it satisfies the following three properties:
    \begin{enumerate}
        \item $\{\al_S\}_{S\subseteq [n]}$ is a canonical representation of an XOS function according to \cref{frac:canon}.
        \item The pointwise maximum of $\{\al_S\}_{S\subseteq [n]}$ is a symmetric function; that is, for every two sets $S,T$ such that $|S|=|T|$, $\al_S(S)=\al_T(T)$.
        \item For every $S \subseteq [n]$, it holds that $\al_S$ is uniform over its support; that is, for every $i\in S$, we have that $\al_S(i)=\frac{1}{|S|}\cdot\al_S(S)$.
    \end{enumerate}
\end{definition}

By Property 1 of the above definition, the XOS canonical representation properties in \cref{lem:canon-rep-properties} hold for canonical set representations as well.

\begin{lemma}
    The class of sXOS functions has a one-to-one correspondence with the set of all canonical set representations.
\end{lemma}

\begin{proof}

    For the first direction, given a canonical set representation $\{\al_S\}_{S\subseteq [n]}$, since it is a canonical representation, the function $f(S)=\max_S\al(S)=\al_S(S)$ is XOS, moreover, since $\al_S(S)=\al_T(T)$ for every $S, T$ of the same size, $f$ is symmetric and thus an sXOS function.
    
    For the other direction: Given some sXOS function $f$, let $\{\al_S\}_{S\subseteq [n]}$ be its XOS canonical representation. We define a new set of uniform additive functions $\{\be_S\}_{S\subseteq [n]}$, where $\be_T(S):=|T\cap S|\cdot\frac{f(|T|)}{|T|}=\frac{|T\cap S|}{|T|}\cdot\al_T(T)$ (the equality holds by symmetry of $f$), and $\be_\emptyset\equiv 0$.
    By construction, every $\be_S$ is uniform over $S$, i.e., $\be_S(i)=\frac{1}{|S|}\cdot\be_S(S)$ for every $i\in S$. We claim that $\{\be_S\}_{S\subseteq [n]}$ is also an XOS canonical representation of $f$, which will conclude the proof. 
    
    Since $\be_S(S)=\al_S(S)$ for every $S$, in order to show that $\be_S(S)= \max_T\be_T(S)$ it is sufficient to show that $\al_S(S)\ge \be_T(S)$ for every $S$ and $T$ (with equality if $S=T$). 
    We proceed by case analysis.
    
    If $|S|\ge|T|$: Let $S_T\subseteq S$ be a subset of size $|T|$. Then it holds that
    \begin{align*}
        \al_S(S)\ge \al_{S_T}(S_T)=\al_T(T)=\be_T(T)\ge \be_T(S),
    \end{align*}
    where the two inequalities follow from the XOS canonical representation properties (see \cref{lem:canon-rep-properties}),
    the first equality is since $f$ is symmetric, and the second equality is from the definition of $\be$. 
    
    If $|S|<|T|$:  
    By Property 3 of the canonical representation of XOS (see \cref{lem:canon-rep-properties}), for every set $T$ there exists an element $i
\in T$ for which 
$$
\frac{\al_T(T)}{|T|}\le \frac{\al_{T\setminus i}(T\setminus i)}{|T\setminus i|}. 
$$ 
We conclude that for every size $k\le |T|$ there exists a set $T_k\subseteq T$ of size $k$ such that 
$$
\frac{\al_T(T)}{|T|}\le \frac{\al_{T_k}(T_k)}{k}.
$$
Therefore we have: 
\begin{align*}
    &\be_T(S) = \frac{|T\cap S|}{|T|}\cdot\al_T(T)
    \le |S|\cdot\frac{\al_T(T)}{|T|}
    \le |S|\cdot\frac{\al_{T_{|S|}}(T_{|S|})}{|T_{|S|}|}
    = \al_S(S)
\end{align*}
where the last equality is since $S$ and $T_{|S|}$ are of the same size.
\end{proof}

We proceed to define two additional representations of sXOS: the canonical \emph{size} and the canonical \emph{function} representations. We show that all three representations are equivalent. Since the canonical function representation is arguably the simplest, we refer to sXOS functions by this representation.

\begin{definition}[sXOS canonical size representation]
    A family of functions $\{\al_k\}_{k\in [n]}$\gcm{,} where $\al_k:[n]\to \mathbb{R}_+$, is a \textit{canonical \emph{size} representation} if it satisfies the following three properties:
    \begin{enumerate}
        \item $\al_k(m)=min\{k, m\}\cdot \al_k(1)$;
        \item For every $k<n$ it holds that $\al_k(k)\le \al_{k+1}(k+1)$;
        \item For every $k>1$ it holds that $\al_k(1)\le \al_{k-1}(1)$.
    \end{enumerate}
\end{definition}

\begin{definition}[sXOS canonical function representation] 
\label{sXOS:main:prop}
    A function $v:[n]\to \mathbb{R}$ is called a \textit{canonical \emph{function} representation} if the following three properties hold:
    \begin{enumerate}
        \item $v(0)=0$;
        \item For every $k<n$ it holds that $v(k)\le v(k+1)$;
        \item For every $k>1$ it holds that $\frac{1}{k}\cdot v(k)\le \frac{1}{k-1}\cdot v(k-1)$.
    \end{enumerate}
\end{definition}

\begin{lemma}[All three sXOS representations are equivalent]
\label{theroem:sXOS}
   The families of canonical set representation, canonical size representation, and canonical function representation have a one-to-one correspondence. 
\end{lemma}
\begin{proof}
    We show how to construct a \emph{size} representation given a \emph{set} representation, a \emph{function} representation given a \emph{size} representation, and finally a \emph{set} representation from a \emph{function} representation. Since we already showed that the family of \emph{set} representations are equivalent to the family of sXOS functions (see \cref{lem-canon-equiv}), we conclude that these are all equivalent.
    
    \begin{itemize}
    \item \textbf{\emph{Set} representation to \emph{size} representation:}
    Given a set representation $\{\al_S\}_{S\subseteq[n]}$ we construct the size representation $\{\be_k\}_{k\in [n]}$ as follows:
    $\be_k(m)=\al_S(T)$ for some two set $S, T$ of sized $|S|=k, |T|=m$ where $S\subseteq T$ if $k\le m$ and $T\subseteq S$ otherwise.
    Properties 1 and 2 are clearly satisfied. As for the third property, let $|S|=k$ and $i\in S$. Then,
    \begin{align*}
        \al_k(1)&=\frac{1}{k}\cdot \al_k(k)=\frac{1}{|S|}\cdot \be_S(S)
        \\&\le \frac{1}{|S|-1}\cdot \be_{S\setminus i}(S\setminus i)
        \\&=\frac{1}{k-1}\cdot \be_{k-1}(k-1)
        =\frac{k-1}{k-1}\cdot \be_{k-1}(1)=\frac{1}{k-1}\cdot \be_{k-1}(1)
    \end{align*}
    
    \item \textbf{\emph{Size} representation to \emph{function} representation:}
    Given a size representation $\{\be_k\}_{k\in [n]}$, let $v$ be a function such that $v(k)=\be_k(k)$ for every $k\in [n]$.
    Properties 1 and 2 are again immediate. For the third property, we have
    \begin{align*}
        \frac{1}{k}\cdot v(k)&=\frac{1}{k}\cdot \be_k(k)=\be_k(1)
        \\&\le \be_{k-1}(1)
        \\&=\frac{1}{k-1}\cdot \be_{k-1}(k-1)
        =\frac{1}{k-1}\cdot v(k-1) 
    \end{align*}

    \item \textbf{\emph{Function} representation to \emph{set} representation:}
    Given a function representation $v:[n]\to\mathbb{R}_+$ let us construct the set representation $\{\al_T\}_{T\subseteq [n]}$ such that $\al_T(S)=\frac{|T\cap S|}{|T|}v(|T|)$.
    We need to show that $\{\al_T\}_{T\subseteq[n]}$ is canonical, uniform, and symmetric.
    Showing uniformity is straightforward, and symmetry is by definition. We show that $\{\al_T\}_{T\subseteq[n]}$ is canonical. Specifically, we show the first property, that $\max_T \al_T(S)=\al_S(S)$.
    
    If $|S|\ge |T|$:
    \begin{align*}
        \al_T(S)&=\frac{|T\cap S|}{|T|}\cdot v(|T|)
        \\&\le \frac{|T|}{|T|}\cdot v(|T|)=v(|T|)
        \\&\le v(|T|+1)\le v(|T|+2)\le  ...\le v(|S|)
        \\&=\frac{|S|}{|S|}\cdot v(|S|)=\frac{|S\cap S|}{|S|}\cdot v(|S|)=\al_S(S).
    \end{align*}
    
    If $|S|\le |T|$:
    \begin{align*}
        \al_T(S)&=\frac{|T\cap S|}{|T|}\cdot v(|T|)
        \\&\le \frac{|S|}{|T|}\cdot v(|T|) = |S|\cdot \frac{v(T)}{|T|}
        \\&\le |S|\cdot \frac{v(|T|-1)}{|T|-1} \le |S| \cdot \frac{v(T-2)}{|T|-2}\le ... \le |S|\cdot \frac{v(|S|)}{|S|}
        \\&=\frac{|S\cap S|}{|S|}\cdot v(|S|)=\al_S(S).
    \end{align*}
    \end{itemize}
\end{proof}

\subsection{XOS and sXOS: Additional Properties}
\label{sub:XOS-sXOS}

In this subsection, we prove some additional properties of sXOS functions. 
These properties are used throughout the paper.

\begin{observation}[Marginals property] \label{sXOS:marginals}
    For any sXOS function $f$ and any positive number $k$, it holds that $m(k)\le \frac{1}{k}f(k)$.
\end{observation}

\begin{proof}
    From the third property of sXOS, we have that
    $$
    f(k-1)\ge \frac{k-1}{k}\cdot f(k).
    $$
    Thus, it holds that 
    $$
    m(k)=f(k)-f(k-1)\le f(k)-\frac{k-1}{k}\cdot f(k)=\frac{1}{k}\cdot f(k).
    $$
\end{proof}

\begin{lemma}[Concavity property]\label{sXOS:concave}
        For every $a\in [0,1]$ 
        for which $a k\in [k]$ it holds that $a \cdot f(k)\le f(a\cdot k)$.
\end{lemma}
\begin{proof}
    Let $m\in [k]$ such that $k-m=a\cdot k$. We have that
    \begin{alignat*}{2}
        \frac{1}{k}\cdot f(k)&\le \frac{1}{k-1}\cdot f(k-1)\le \frac{1}{k-2}\cdot f(k-2)\le \cdots &&\le \frac{1}{k-m}\cdot f(k-m)
        = \frac{1}{\aL \cdot k}\cdot f(\aL \cdot k).
    \end{alignat*}
\end{proof}

\paragraph{Extending these properties to general XOS.}
\label{extending:sXOS:prop}

Observe that the properties we have shown for sXOS can be weakly extended to the general XOS function.
Recall that the third property of canonical XOS is that for every set $S$, letting $k=|S|$, there exists a set $i\in S$ such that $\frac{1}{k}f(k)\le \frac{1}{k-1}f(S\setminus i)$. 
This is equivalent to $f(i:S)\le \frac{1}{|S\setminus i|}f(S\setminus i)$ and by the same argument as the sXOS marginal property, for this \emph{specific} agent $i$, it holds that $f(i:S)\le \frac{1}{|S|}f(S)$.
Similarly, again with a similar argument as shown for sXOS, we have that for every $\aL\in [0,1]$ such that $\aL \cdot k\in [n]$ there exists a subset $T\subseteq S$ of size $\aL$ such that $\aL \cdot f(T)\le f(S)$.

%% file: KKB_algorithm2.tex
\section[Proof of Lemma~\ref{lemma:alg2}]{Proof of \cref{lemma:alg2}}
\label{section:XOS:part2}
\label{appendix:XOS}

In this section, we prove \cref{lemma:alg2}, which asserts the existence and properties of \cref{alg2} used in our XOS welfare and value approximation algorithms. The following notation will be used throughout the proof: 
\begin{equation}
    \psi(f, y):=\left(1-\frac{1}{M(f, y)}\right)\cdot \frac{A}{m}\cdot y=\frac{A}{m}\cdot y-x(f).\label{eq:psi}
\end{equation}

The proof of \cref{lemma:alg2} relies on two lemmas: The first is Lemma~\ref{alg1:usage} (see Section~\ref{appx:alg1-analysis}), which establishes the guarantees of \cref{alg1} utilized by \cref{alg2}. The second is Lemma~\ref{L:3.3} and Corollary~\ref{cor:y-kappa-bounded} (see Section~\ref{appx:k-bound}), where the latter establishes a $\kappa$-boundedness property guaranteed by the conditions of \cref{lemma:alg2}. In Section~\ref{appx:alg2-and-proof} we present \cref{alg2} itself and put everything together to complete the proof of \cref{lemma:alg2}.
Section~\ref{appendix:alg1} appears for completeness and contains \cref{alg1} from~\cite{multi-agents} and two useful results of~\cite{multi-agents} (Lemmas~\ref{lem:prop-of-XOS}-\ref{L:2.1}).

\subsection{Lemma~\ref{alg1:usage}: Guarantees of Algorithm~\ref{alg1}}
\label{appx:alg1-analysis}

\cref{alg1} returns a subset $U \subseteq [n]$ with value in the range of a given target value $\Psi$ while guaranteeing that the marginals for $U$ are not too small. 
The following \cref{alg1:usage} shows how to use \cref{alg1} to find a set $U$ satisfying the guarantees of \cref{lemma:alg2}. However, the validity requirements in \cref{alg1:usage} are not the same as the ones of \cref{lemma:alg2}. In Section~\ref{appx:alg2-and-proof} we show how the validity requirements in \cref{lemma:alg2} imply the validity requirements of \cref{alg1:usage}. 

To state the lemma we use the following definition:

\begin{definition}[$h$-minimal and welfare minimal set]\label{def:welfare-mininal}
    For any set function $h$ with domain $2^{[n]}$, set $S$ is \emph{$h$-minimal} if $h(S)\ge h([n])$, and $h(i:S)> 0$ for every $i\in S$. 
    Set $S$ is \emph{welfare minimal} if it is a $w$-minimal set for the welfare function~$w$.
    Set $S$ is an \emph{$h$-minimal subset of $T$} if $S$ is an $h_{|T}$-minimal set.
\end{definition}

Note that an $h$-minimal set can be computed in $O(n^2)$ time by iteratively removing agents with non-positive marginals.

\begin{lemma}
\label{alg1:usage}
    Let $f:2^{[n]}\to\mathbb{R}_+$ be an XOS function, and $c:2^{[n]}\to\mathbb{R}_+$ an additive cost function. 
    Let  $y\ge 0$, $m\ge 1$, $a\in (0, 1]$, and for every $i$, let $\tilde c_i=2\sqrt{c_i Ay}$. 
    Let $Z$ be a welfare minimal set with respect to $f$ and $\tilde c$. 
    If 
    $\psi(f,y) < f(Z)$, then \cref{alg1} called with $f_{|Z}$ and $\psi(f,y)$ will return a set $U$ that satisfies:
    \begin{enumerate}[nosep,noitemsep]
        \item $\frac{A}{m}\cdot y\ge f(U)\ge \left(1-\frac{1}{M(f,y)}\right)\cdot \frac{A}{2m}\cdot y$;
        \item $\rho(U)\le \frac{1}{m}$.
    \end{enumerate}
\end{lemma}

\begin{proof}
Note that $x(f_{|Z})\le x(f)$. By definition of $\psi(f,y)$, the set $U\subseteq Z$ returned by Algorithm 1 has the following three properties:
\begin{enumerate}[nosep,noitemsep]
    \item $f(U)\ge \frac{1}{2}\psi(f, y) = (1-\frac{1}{M(f, y)})\cdot\frac{A}{2m}\cdot y$
    \item $f(U)\le \psi(f,y) + x(f_{|Z}) = \frac{A}{m}\cdot y-x(f) + x(f_{|Z}) \le \frac{A}{m}\cdot y$
    \item $f(i:U)\ge \frac{1}{2}f(i:Z)$ for every agent $i\in U$
\end{enumerate}
The first guarantee follows from the first and second property.
As for the second property, notice that 
\[
    \tilde c_i \; = \; 2\sqrt{c_i Ay} \; = \; 2\sqrt{c_im\frac{A}{m}y} \; \ge \; 2\sqrt{c_im f(U)}
\]
where the last inequality is from the second property. Thus
\[
    f(i:U) \; \ge \; \frac{1}{2}f(i:Z) 
    \; \ge \; \frac{1}{2}\tilde{c}_i 
    \; \ge \; \sqrt{c_i m f(U)}
\]
where the first inequality is the third property of $U$. The second inequality holds since $Z$ is a welfare minimal set of $f$ and $\tilde c$, and since $\tilde c$ is additive. By squaring both sides and rearranging, we obtain
\[
    \frac{f(i:U)}{m f(U)} \; \ge \; \frac{c_i}{f(i:U)}\enspace.
\]
Using the sum of marginals property, this leads to
\[
    \rho(U) 
    \; = \; \sum_{i\in U}\frac{c_i}{f(i:U)}
    \; \le \; \frac{1}{m f(U)}\sum_{i\in U}f(i:U)
    \; \le \; \frac{1}{m}\enspace. \qedhere
\]
\end{proof}

\subsection[Corollary~\ref{cor:y-kappa-bounded}: Kappa-Boundedness]{Corollary~\ref{cor:y-kappa-bounded}: $\kappa$-Boundedness}
\label{appx:k-bound}

In this section we define and discuss \emph{$\kappa$-boundedness}. Our motivation is as follows.
In order to use \cref{alg1:usage} to obtain a set $U$ with the desired properties, we need to choose $y$ and $Z$ such that $\left(f_{|Z}, \psi(f, y)\right)$ is a valid input for Algorithm~1, i.e., such that 
$0\le \psi(f, y)\le f(Z)$.
The non-negativity of $\psi(f, y)$ holds, since for this it is sufficient that $M(f, y) \ge 1$ (see Equation~\eqref{eq:psi}), which holds by Condition~(3) of Lemma~\ref{lemma:alg2}. We thus seek sufficient conditions for the following inequality:
\begin{equation}
    \psi(f, y)\le f(Z).\label{eq:innocent}
\end{equation}
We develop such conditions in \cref{N:3.6} (Section~\ref{appx:alg2-and-proof}), which rely on $\kappa$-boundedness.

\begin{definition}[$\kappa$-boundedness]
\label{def:k-bounded}
    Consider an $n$-agent contract setting $(f,c)$.
    For a given $\kappa \ge 0$, a value $y$ is $\kappa$-bounded with respect to $(f,c)$ if there exists a set $S\subseteq [n]$ such that (1)~$y\le f(S)$, and (2)~$\,\frac{\sum_{i\in S}\sqrt{c_i} }{\sqrt {f(S)}}\le \kappa$. 
    In this case we say that $y$ is $\kappa$-bounded by $S$.
\end{definition}

The next lemma shows a sufficient condition for a set to fulfill the second condition of $\kappa$-boundedness:
there exists an XOS function $h$ and a superset $T$ of $S$ for which $S$'s share in $\rho_T^{(h)}$ is upper bounded by $\kappa^2$. 
Note that $T$ can even be a superset of $[n]$ since the cost of any agent $i\notin S$ is not used.

\begin{lemma} \label{L:3.3}
    Let $(f,c)$ be an $n$-agent contract setting with an XOS value function and let $S\subseteq[n]$.
    If there exists an XOS domain extension $h$ of $f_{|S}$ with a domain $2^T$ such that $\rho_T^{(h)}(S)\le \kappa^2$, 
    then $\frac{\sum_{i\in S}\sqrt{c_i}}{\sqrt{f(S)}}\le \kappa$.
\end{lemma}

\begin{proof}
\begin{align*}
    \left(\; \sum_{i\in S}\sqrt{c_i} \;\right)^2
    &=\left(\; \sum_{i\in S}\sqrt{\frac{c_i}{f(i:T)}}\cdot\sqrt{f(i:T)} \;\right)^2 
    = \left(\; \sum_{i\in S}\sqrt{\rho_T^{(h)}(i)}\cdot\sqrt{f(i:T)} \;\right)^2
    \\&\le \left(\; {\sum_{i\in S}\rho^{(h)}_T(i)} \right)\left(\; \sum_{i\in S}f(i:T) \;\right) 
    = \rho^{(h)}_T(S)\cdot\sum_{i\in S}f(i:T) 
    \\&\le \rho_T^{(h)}(S)\cdot f(S),
\end{align*}
where the first inequality follows from Cauchy–Schwartz, and the second inequality is from the sum of marginals property of XOS functions (see \cref{L:2.1}).
\end{proof}

\begin{corollary} 
    \label{cor:y-kappa-bounded}
    Let $(f,c,y,m)$ be an input which is \emph{valid} according to the conditions in Lemma~\ref{lemma:alg2}. Then $y$ is $\kappa$-bounded with respect to $(f,c)$ for $\kappa=2-\frac{1}{8m}$.
\end{corollary}

\begin{proof}
    By Condition~(2) on the input (see Lemma~\ref{lemma:alg2}), and by applying \cref{L:3.3}, there exists a set $S$ such that the second requirement of $\kappa$-boundedness (see Definition~\ref{def:k-bounded}) holds for $S$ and $\kappa=2-\frac{1}{8m}$. Since by Condition~(1) on the input, $y\le f(S)$, the first requirement of $\kappa$-boundedness holds as well, and we conclude that $y$ is $\kappa$-bounded.
\end{proof}

\subsection{Algorithm~\ref{alg2} and Putting Everything Together}
\label{appx:alg2-and-proof}

\setcounter{algocf}{1}

\begin{algorithm}[ht]
\DontPrintSemicolon
\caption{Routine to attempt a valid call to Algorithm~1}
\label{alg2}
\KwInput{An XOS function $f:2^{[n]}\to\mathbb{R}_+$, an additive cost function $c:2^{[n]}\to\mathbb{R}_+$, a value $y>0$, $m\ge 1$, and an $a$-demand oracle.}
\KwOutput{Either $U = \emptyset$, or $U$ satisfying the properties of \cref{alg1:usage}}
\medskip
Let $\tilde{c}$ be an additive function such that $\tilde{c}_i = \sqrt{c_i m A y}$\;
Let $D$ be an $a$-demand set for $(f,\tilde c)$, and $Z$ be a welfare minimal subset of $D$.\;
\bfIf\ \emph{$0\le \psi(f, y)< f(Z)$} \bfThen\ \Return output of Algorithm 1 on $\left(f_{|Z}, \psi(f,y)\right)$ \bfElse\ \Return\ $\emptyset$
\end{algorithm}

Consider~\cref{alg2}.
Observe that if a set $S$ is an $a$-demand set, then so is any of its welfare-minimal subsets:
By definition, a set $S$ is $a$-demand if for any other set $T$ it holds that $w(S)\ge a\cdot f(T)-c(T)$. Let $S'$ be a minimal welfare set of $S$, then $w(S')\ge w(S)$; thus, it is also an $a$-demand set.
For this reason we get that $Z$ is an $a$-demand set with respect to $f$ and $\tilde c$.

The next lemma shows that since $Z$ is an $a$-demand set, then if $y$ is $\kappa$-bounded (with $\kappa$ sufficiently small), Equation \eqref{eq:innocent} holds and thus the guarantees of \cref{alg1:usage} apply to the output of \cref{alg1} called from \cref{alg2}. This lemma is then sufficient to complete the proof of \cref{lemma:alg2}, by applying~\cref{cor:y-kappa-bounded} to show that $y$ is $\kappa$-bounded.

\begin{lemma} \label{N:3.6}
    Let f be an XOS function and $\tilde c$ be as defined in \cref{alg1:usage}.
    Let $Z$ be an $a$-demand set of $f$ and $\tilde c$.
    If $y$ is $\kappa$-bounded with respect to $f$ and $c$ with $\kappa\le \left(2-\frac{1}{8m}\right)$, then $\psi(f, y)\le f(Z)$.
\end{lemma}
\begin{proof}
   Recall that $A=(a/4)^2$, and thus $\tilde c_i=2\sqrt{c_iAy}=\frac{a}{2}\sqrt{c_i y}$. 
    Let $S$ be such that $y$ is $\kappa$-bounded by $S$, then

    \begin{align*}
        f(T)&\ge f(T)-\tilde{c}(T) 
        \\&\ge \aL f(S)-\tilde{c}(S)
          = \aL f(S) - \frac{\aL}{2}\sqrt{y}\cdot\sum_{i\in S} \sqrt{c_i} && \text{($T$ is an $a$-demand set for $\tilde{c}$)}
        \\&\ge \aL f(S) - \frac{\aL}{2}\sqrt{y}\cdot \kappa \cdot \sqrt{f(S)} && \text{($\kappa$-boundedness)}
        \\&\ge \aL f(S)-\aL \cdot \frac{\kappa}{2}\cdot f(S) = \aL \left(1-\frac{\kappa}{2}\right)f(S)
        \\&\ge \; \frac{\aL}{16m} \cdot f(S) \;
          \ge \; \frac{\aL}{16m}\cdot y \; > \; \Psi && \text{($\kappa \le \left(2- \frac{1}{8m}\right)$, and $\aL \le 1$)} \qedhere
    \end{align*}
    
\end{proof}

We are ready to prove \cref{lemma:alg2}.

\begin{proof}[Proof of \cref{lemma:alg2}]  
    From \cref{cor:y-kappa-bounded}, we know that $y$ is $\kappa$-bounded with $\kappa\le \left(2-\frac{1}{8m}\right)$.
    Thus from \cref{N:3.6}, we know that since $Z$ is an $a$-demand set, Equation~\eqref{eq:innocent} holds.
    Since it also holds that $M(f, y)\ge 1$, we have that $\psi(f, y)\ge 0$. This shows that $(f_{|Z}, \psi(f, y))$ is a valid input to \cref{alg1}.
    Thus, by \cref{alg1:usage}, when running \cref{alg1} with the input $(f_{|Z}, \psi(f, y))$, the output $U$ satisfies $\rho(U)\le \frac{1}{m}$ and $\frac{A}{m}\cdot y \ge f(U)\ge \left(1-\frac{1}{M(f, y)}\right)\cdot\frac{A}{2m}\cdot y$, as claimed.
\end{proof}

\subsection{For Completeness: Algorithm \ref{alg1} and 
Lemmas~\ref{lem:prop-of-XOS}-\ref{L:2.1}}
\label{appendix:alg1}
 
In this section we state for completeness two useful properties of XOS functions shown in \cite{multi-agents}.

\begin{lemma}[Scaling sets property {\cite[Lem.~3.5]{multi-agents}}]
\label{lem:prop-of-XOS}
    For every XOS function $f:2^{[n]}\to \mathbb R_+$ and target value $\Psi \in [0, f([n])]$, there exists a set $U$ which satisfies:
    \begin{itemize}[nosep,noitemsep]
        \item $\frac{1}{2}\Psi \leq f(U) \leq \Psi + x(f)$;
        \item $f(i : U) \geq \frac{1}{2} f(i : [n]) \;\, \text{for all } i \in U$.
    \end{itemize}
    Additionally, there exists an efficient algorithm that finds the set $U$ given value oracle access to $f$.
\end{lemma}

We refer to the algorithm guaranteed by \cref{lem:prop-of-XOS} as \cref{alg1}, and include its description for completeness. 

\setcounter{algocf}{0}
\begin{algorithm}[ht]
\caption{Adapted from \cite[Alg.~1]{multi-agents}}
\label{alg1}
\DontPrintSemicolon
\KwInput{Value oracle access to an XOS function $f : 2^{[n]} \rightarrow \mathbb{R}_+$, a parameter $0 \leq \Psi \le f([n])$}
\KwOutput{A set $U$ which satisfies the conditions of \cref{alg1:usage}}
\medskip

\bfIf{ $\Psi=f([n])$ \Return [n]}

\textbf{let} $T_0$ be an $f$-minimal set (Definition~\ref{def:welfare-mininal})

\For{$t = 1,\ldots,|T_0|$}{
$i_t \gets \arg \min_{i \in T_{t-1}} f(i:\,T_{t-1})/f(i:T_0)$ \;
$T_t \gets (T_{t-1}\setminus \{i_t\})$\;
$\delta_t \gets f(i_t:T_t)/f(i_t:T_0)$\;
}
$i^{(-)} \gets \min\{t : f(T_t) \leq \Psi\}$\;
$i^{(+)} \gets \min\{t : f(T_t) \leq \frac{1}{2}f(T_{i^{(-)}-1})\}$\;
$i^{\star} \gets \arg\max\{\delta_t: t = i^{(-)}, \ldots, i^{(+)}\}$\;

\Return $U = T_{i^{\star}-1}$
\end{algorithm}

\begin{lemma}[XOS sum of marginals {\cite[Lem~2.1]{multi-agents}}]\label{L:2.1}
    For any XOS function $f$ and any two sets $S\subseteq T$, it holds that $\sum_{i\in S}f(i:T)\le f(S)$.
\end{lemma}

The following proof of \cref{L:2.1} uses the properties of the canonical representation for XOS functions (see \cref{frac:canon}).

\begin{proof}[Proof of \cref{L:2.1}]
    \begin{align*}
        \sum_{i\in S}f(i:T) &
        = \sum_{i\in S} (\al_T(T)-\al_{T\setminus \{i\}}(T\setminus \{i\}))
        \\ &\le \sum_{i\in S} (\al_T(T)-\al_T(T\setminus \{i\}))
        \\&= \sum_{i\in S} \al_T(i) 
        = \al_T(S)
        \le \al_S(S) = f(S),
    \end{align*}
    where the last inequality is since for every two sets $S$ and $T$ we have that $\al_S(S)\ge \al_T(S)$.
\end{proof}

%% file: AA_main.bbl
\begin{thebibliography}{10}

\bibitem{AlonLST23}
T.~Alon, R.~Lavi, E.~S. Shamash, and I.~Talgam{-}Cohen.
\newblock Incomplete information {VCG} contracts for common agency.
\newblock {\em Operations Research}, pages 288--299, 2024.

\bibitem{original-contract}
M.~Babaioff, M.~Feldman, N.~Nisan, and E.~Winter.
\newblock Combinatorial agency.
\newblock {\em Journal of Economic Theory}, 147(3):999--1034, 2012.

\bibitem{BalmacedaBCS16}
F.~Balmaceda, S.~R. Balseiro, J.~R. Correa, and N.~E. Stier-Moses.
\newblock Bounds on the welfare loss from moral hazard with limited liability.
\newblock {\em Games and Economic Behaviour}, 95:137--155, 2016.

\bibitem{medicare}
H.~Bastani, M.~Bayati, M.~Braverman, R.~Gummadi, and R.~Johari.
\newblock Analysis of medicare pay-for-performance contracts.
\newblock In {\em Mechanism Design for Social Good Workshop}, 2017.
\newblock Available at \url{https://papers.ssrn.com/sol3/papers.cfm?abstract_id=2839143}.

\bibitem{CacciamaniEtAl24}
F.~Cacciamani, M.~Bernascon, M.~Castiglioni, and N.~Gatti.
\newblock Multi-agent contract design beyond binary actions.
\newblock In {\em EC 2024}, 2024.

\bibitem{hidden-contracts}
M.~Castiglioni, A.~Marchesi, and N.~Gatti.
\newblock Multi-agent contract design: How to commission multiple agents with individual outcome.
\newblock In {\em EC 2023}, pages 412--448, 2023.

\bibitem{ChristoffersenH23}
P.~J.~K. Christoffersen, A.~A. Haupt, and D.~Hadfield{-}Menell.
\newblock Get it in writing: Formal contracts mitigate social dilemmas in multi-agent {RL}.
\newblock In {\em Proceedings of the 2023 International Conference on Autonomous Agents and Multiagent Systems {AAMAS}}, pages 448--456. {ACM}, 2023.

\bibitem{Clarke71}
E.~H. Clarke.
\newblock Multipart pricing of public goods.
\newblock {\em Public Choice}, 11(1):17--33, 1971.

\bibitem{deforestation2}
W.~Dai~Li, I.~Ashlagi, and I.~Lo.
\newblock Simple and approximately optimal contracts for payment for ecosystem services.
\newblock {\em Management Science}, 69(12):7821--7837, 2022.

\bibitem{supermodular}
R.~Deo{-}Campo~Vuong, S.~Dughmi, N.~Patel, and A.~Prasad.
\newblock On supermodular contracts and dense subgraphs.
\newblock In {\em SODA 2024}, pages 109--132, 2024.

\bibitem{DobzinskiN10}
S.~Dobzinski and N.~Nisan.
\newblock Mechanisms for multi-unit auctions.
\newblock {\em J. Artif. Intell. Res.}, 37:85--98, 2010.

\bibitem{DuettingEFK21}
P.~D\"utting, T.~Ezra, M.~Feldman, and T.~Kesselheim.
\newblock Combinatorial contracts.
\newblock In {\em FOCS 2021}, pages 815--826, 2021.

\bibitem{multi-agents}
P.~D\"utting, T.~Ezra, M.~Feldman, and T.~Kesselheim.
\newblock Multi-agent contracts.
\newblock In {\em Proceedings of the 55th Annual {ACM} Symposium on Theory of Computing, {STOC}}, pages 1311--1324, 2023.

\bibitem{multi-multi}
P.~D\"utting, T.~Ezra, M.~Feldman, and T.~Kesselheim.
\newblock Multi-agent combinatorial contracts.
\newblock In {\em SODA 2025}, 2025.
\newblock Forthcoming.

\bibitem{dutting2024algorithmic}
P.~D{\"u}tting, M.~Feldman, and I.~Talgam-Cohen.
\newblock Algorithmic contract theory: A survey.
\newblock {\em Foundations and Trends in Theoretical Computer Science}, 16(3-4):211--412, 2024.

\bibitem{minmax1}
P.~D\"utting, T.~Roughgarden, and I.~Talgam{-}Cohen.
\newblock Simple versus optimal contracts.
\newblock In {\em EC 2019}, pages 369--387, 2019.
\newblock Full version available at \url{https://arxiv.org/pdf/1808.03713}.

\bibitem{DuttingRT20}
P.~D\"utting, T.~Roughgarden, and I.~Talgam{-}Cohen.
\newblock The complexity of contracts.
\newblock In {\em SODA 2020}, pages 2688--2707, 2020.

\bibitem{inapproximability}
T.~Ezra, M.~Feldman, and M.~Schlesinger.
\newblock On the (in)approximability of combinatorial contracts.
\newblock In {\em ITCS 2024}, pages 44:1--44:22, 2024.

\bibitem{inapproximability:prop}
U.~Feige.
\newblock A threshold of ln n for approximating set cover.
\newblock {\em J. ACM}, 45(4):634–652, July 1998.

\bibitem{feldman2025budget}
M.~Feldman, Y.~Gal-Tzur, T.~Ponitka, and M.~Schlesinger.
\newblock Budget-feasible contracts.
\newblock Working paper available at \url{https://arxiv.org/pdf/2504.01773}, 2025.

\bibitem{budget-contracts}
S.~Goel and W.~Hann-Caruthers.
\newblock Optimality of weighted contracts for multi-agent contract design with a budget.
\newblock In {\em EC 2024}, 2024.

\bibitem{grossman1992analysis}
S.~J. Grossman and O.~D. Hart.
\newblock An analysis of the principal-agent problem.
\newblock In {\em Foundations of Insurance Economics: Readings in Economics and Finance}, pages 302--340. Springer, 1992.

\bibitem{Groves73}
T.~Groves.
\newblock Incentives in teams.
\newblock {\em Econometrica}, 41(4):617--631, 1973.

\bibitem{submodular-demand2}
C.~Harshaw, M.~Feldman, J.~Ward, and A.~Karbasi.
\newblock Submodular maximization beyond non-negativity: Guarantees, fast algorithms, and applications, 2019.

\bibitem{crowdsource}
C.~Ho, A.~Slivkins, and J.~W. Vaughan.
\newblock Adaptive contract design for crowdsourcing markets: Bandit algorithms for repeated principal-agent problems.
\newblock {\em Journal of Artificial Intelligence Research}, 55:317--359, 2016.

\bibitem{holmstrom1979moral}
B.~Holmstr{\"o}m.
\newblock Moral hazard and observability.
\newblock {\em The Bell journal of economics}, pages 74--91, 1979.

\bibitem{ivanov2024principal}
D.~Ivanov, P.~D{\"u}tting, I.~Talgam-Cohen, T.~Wang, and D.~C. Parkes.
\newblock Principal-agent reinforcement learning: Orchestrating {AI} agents with contracts.
\newblock Working paper available at \url{https://arxiv.org/abs/2407.18074}, 2024.

\bibitem{KangM22}
K.~Kang and R.~A. Miller.
\newblock Winning by default: Why is there so little competition in government procurement?
\newblock {\em The Review of Economic Studies}, 89:1495--1556, 2022.

\bibitem{LehmannLN06}
B.~Lehmann, D.~Lehmann, and N.~Nisan.
\newblock Combinatorial auctions with decreasing marginal utilities.
\newblock {\em Games and Economic Behavior}, 55(2):270--296, 2006.

\bibitem{Hierarchy1}
B.~Lehmann, D.~Lehmann, and N.~Nisan.
\newblock Combinatorial auctions with decreasing marginal utilities.
\newblock {\em Games and Economic Behavior}, 55(2):270--296, 2006.

\bibitem{LiIL21}
W.~D. Li, N.~Immorlica, and B.~Lucier.
\newblock Contract design for afforestation programs.
\newblock In {\em WINE 2021}, pages 113--130, 2021.

\bibitem{Myerson81}
R.~B. Myerson.
\newblock Optimal auction design.
\newblock {\em Mathematics of Operations Research}, 6(1):58--73, 1981.

\bibitem{nisan2015algorithmic}
N.~Nisan.
\newblock Algorithmic mechanism design: Through the lens of multiunit auctions.
\newblock In {\em Handbook of Game Theory with Economic Applications}, volume~4, pages 477--515. Elsevier, 2015.

\bibitem{Nobel}
nobelprize.org.
\newblock Advanced information.
\newblock \url{https://nobelprize.org/prizes/economic-sciences/2016/advanced-information/}, 2016.

\bibitem{SaigTR23}
E.~Saig, I.~Talgam{-}Cohen, and N.~Rosenfeld.
\newblock Delegated classification.
\newblock In {\em {NeurIPS} 2023}, 2023.

\bibitem{socialmedia}
{Statista Research Department}.
\newblock Instagram influencer marketing spending worldwide from 2013 to 2020.
\newblock \url{https://www.statista.com/statistics/950920/global-instagram-influencer-marketing-spending/}, 2021.

\bibitem{submodular-demand1}
M.~Sviridenko, J.~Vondrák, and J.~Ward.
\newblock Optimal approximation for submodular and supermodular optimization with bounded curvature, 2014.

\bibitem{Vickrey61}
W.~Vickrey.
\newblock Counterspeculation, auctions, and competitive sealed tenders.
\newblock {\em The Journal of Finance}, 16(1):8--37, 1961.

\bibitem{Yocumetal23}
J.~Yocum, P.~Christoffersen, M.~Damani, J.~Svegliato, D.~Hadfield-Menell, and S.~Russell.
\newblock Mitigating generative agent social dilemmas.
\newblock In {\em FMDM@NeurIPS 2023}, 2023.

\end{thebibliography}
